\documentclass[a4paper,UKenglish]{lipics-arxive}

\usepackage[utf8]{inputenc}
\usepackage{amsthm,amsmath,amssymb}
\usepackage{bbm}
\usepackage{misc}
\usepackage{boxedminipage}
\usepackage{pifont}
\usepackage{xspace}
\usepackage{color}
\usepackage{microtype}
\usepackage{stmaryrd}

 \title{Conservative Extensions in Guarded and Two-Variable Fragments%
\footnote{Funded by DFG grant LU 1417/2.}}

\author[1]{Jean Christoph Jung}
\author[1]{Carsten Lutz}
\author[1]{Mauricio Martel}
\author[1]{Thomas Schneider}
\author[2]{Frank Wolter}
\affil[1]{University of Bremen, Bremen, Germany\\
  \texttt{\{jeanjung,clu,martel,ts\}@informatik.uni-bremen.de}}
\affil[2]{University of Liverpool, Liverpool, UK\\
  \texttt{wolter@liverpool.ac.uk}}
\authorrunning{J.\,C. Jung, C. Lutz, M. Martel, T. Schneider,
and F. Wolter} 

\Copyright{J.\,C. Jung, C. Lutz, M. Martel, T. Schneider,
and F. Wolter}

\newcommand{\sig}{\textup{\textsf{sig}}\xspace}

\begin{document}

\maketitle

\begin{abstract}
  We investigate the decidability and computational complexity of
  (deductive) conservative extensions in fragments of first-order
  logic (FO), with a focus on the two-variable fragment FO$^2$ and the
  guarded fragment GF. We prove that conservative extensions are
  undecidable in any FO fragment that contains FO$^2$ or GF (even the
  three-variable fragment thereof), and that they are decidable and
  2\ExpTime-complete in the intersection GF$^2$ of FO$^2$ and
  GF. 
\end{abstract}

\section{Introduction}

Conservative extensions are a fundamental notion in logic. In
mathematical logic, they provide an important tool for relating
logical theories, such as theories of arithmetic and theories that
emerge in set theory \cite{SecondArith,Pollard}. In computer science,
they come up in diverse areas such as software specification
\cite{Goguen}, higher order theorem proving \cite{GM93}, and
ontologies~\cite{DBLP:series/lncs/KonevLWW09}. In these applications,
it would be very useful to decide, given two sentences $\vp_1$ and
$\vp_2$, whether $\vp_1 \wedge \vp_2$ is a conservative extension of
$\vp_1$.  As expected, this problem is undecidable in first-order
logic (FO). In contrast, it has been observed in recent years that
conservative extensions are decidable in many modal and description
logics
\cite{GhilardiLutzWolter-KR06,DBLP:conf/ijcai/LutzWW07,DBLP:conf/ijcai/LutzW11,DBLP:conf/rweb/BotoevaKLRWZ16}. This
observation is particularly interesting from the viewpoint of
ontologies, where conservative extensions have many natural
applications including modularity and reuse, refinement, versioning,
and forgetting~\cite{JairGrau,DBLP:series/lncs/KonevLWW09}.

Regarding decidability, conservative extensions thus seem to behave
similarly to the classical satisfiability problem, which is also
undecidable in FO while it is decidable for modal and description
logics. In the case of satisfiability, the aim to understand the
deeper reasons for this discrepancy and to push the limits of
decidability to more expressive fragments of FO has sparked a long
line of research that identified prominent decidable FO fragments such
as the two-variable fragment FO$^2$ \cite{Scott,DBLP:journals/mlq/Mortimer75}, its extension with
counting quantifiers C$^2$ \cite{GOR97-C2}, the guarded fragment GF \cite{ANvB98},
and the guarded negation fragment GNF \cite{DBLP:journals/jacm/BaranyCS15}, see also
\cite{Borgeretal97,DBLP:journals/jsyml/Gradel99,DBLP:journals/iandc/Pratt-Hartmann09,DBLP:journals/siamcomp/KieronskiMPT14} and 
references therein. These fragments have sometimes been used as a replacement for
the modal and description logics that they generalize, and in
particular the guarded fragment has been proposed as an ontology
language \cite{DBLP:journals/corr/BaranyGO13}. Motivated by this situation, the aim of the
current paper is to study the following two questions:
\begin{enumerate}

\item Are conservative extensions decidable in relevant fragments of
  FO such as FO$^2$, C$^2$, GF, and GNF?

\item What are the deeper reasons for decidability of conservative
  extensions in modal and description logics and how far can the
  limits of decidability be pushed?

\end{enumerate}
To be more precise, we concentrate on \emph{deductive} conservative
extensions, that is, $\vp_1 \wedge \vp_2$ is a conservative extension
of $\vp_1$ if for every sentence $\psi$ formulated in the signature of
$\vp_1$, $\vp_1 \wedge \vp_2 \models \psi$ implies $\vp_1 \models
\psi$. There is also a \emph{model-theoretic} notion of conservative
extension which says that $\vp_1 \wedge \vp_2$ is a conservative
extension of $\vp_1$ if every model of $\vp_1$ can be extended to a
model of $\vp_2$ by interpreting the additional symbols in
$\vp_2$. Model-theoretic conservative extensions imply deductive
conservative extensions, but the converse fails unless one works with
a very expressive logic such as second-order logic
\cite{DBLP:series/lncs/KonevLWW09}. In fact, model-theoretic
conservative extensions are undecidable even for some very
inexpressive description logics that include neither negation nor
disjunction \cite{DBLP:journals/ai/KonevL0W13}. Deductive conservative
extensions, as studied in this paper, are closely related to other
important notions in logic, such as uniform interpolation
\cite{Pitts,Visser,DBLP:conf/lics/BenediktCB15}. For example, in
logics that enjoy Craig interpolation, a decision procedure for
conservative extensions can also be used to decide whether a given
sentence $\vp_2$ is a uniform interpolant of a given sentence $\vp_1$
regarding the symbols used in $\vp_2$.

Instead of concentrating only on conservative extensions, we also
consider two related reasoning problems that we call $\Sigma$-entailment and
$\Sigma$-inseparability, where $\Sigma$ denotes a signature.  The
definitions are as follows: a sentence $\vp_1$ $\Sigma$-entails a
sentence $\vp_2$ if for every sentence $\psi$ formulated in $\Sigma$,
$\vp_2 \models \psi$ implies $\vp_1 \models \psi$.  This can be viewed
as a more relaxed notion of conservative extension where it is not
required that one sentence actually extends the other as in the
conjunction $\vp_1 \wedge \vp_2$ used in the definition of
conservative extensions. Two sentences $\vp_1,\vp_2$ are
$\Sigma$-inseparable if they $\Sigma$-entail each other. We generally
prove lower bounds for conservative extensions and upper bounds for
$\Sigma$-entailment, in this way obtaining the same decidability and
complexity results for all three problems.

Our first main result is that conservative extensions are undecidable
in FO$^2$ and (the three-variable fragment of) GF, and in fact in all
fragments of FO that contain at least one of the two; note that the
latter is not immediate because the separating sentence $\psi$ in the
definition of conservative extensions ranges over all sentences from
the considered fragment, giving greater separating power when we move
to a larger fragment. The proofs are by reductions from the halting
problem for two-register machines and a tiling problem,
respectively. We note that undecidability of conservative extensions
also implies that there is no extension of the logic in question in
which consequence is decidable and that has effective uniform
interpolation (in the sense that uniform interpolants exist and are
computable).  We then show as our second main result that, in the
two-variable guarded fragment GF$^2$, $\Sigma$-entailment is decidable
in 2\ExpTime.  Regarding the satisfiability problem, GF$^2$ behaves
fairly similarly to modal and description logics.  It is thus
suprising that deciding $\Sigma$-entailment (and conservative
extensions) in GF$^2$ turns out to be much more challenging than in
most modal and description logics. There, the main approach to proving
decidability of $\Sigma$-entailment is to first establish a suitable
model-theoretic characterization based on bisimulations which is then
used as a foundation for a decision procedure based on tree automata
\cite{DBLP:conf/ijcai/LutzW11,DBLP:conf/rweb/BotoevaKLRWZ16}. In GF$^2$, an analogous
characterization in terms of appropriate guarded bisimulation
fails. Instead, one has to demand the existence of \emph{$k$-bounded}
(guarded) bisimulations, \emph{for all} $k$, and while tree automata
can easily handle bisimulations, it is not clear how they can deal
with such an infinite family of bounded bisimulations. We solve this
problem by a very careful analysis of the situation and by providing
another characterization that can be viewed as being `half way'
between a model-theoretic characterization and an automata-encoding of
$\Sigma$-entailment.

We also observe that a 2\ExpTime lower bound from
\cite{GhilardiLutzWolter-KR06} for conservative extensions in
description logics can be adapted to GF$^2$, and consequently our
upper bound is tight. It is known that GF$^2$ enjoys Craig
interpolation and thus our results are also relevant to deciding
uniform interpolants and to a stronger version of conservative
extensions in which the separating sentence $\psi$ can also use
`helper symbols' that occur neither in $\vp_1$ nor in $\vp_2$.

\section{Preliminaries}
\label{sect:prelim}
\renewcommand{\vec}[1]{\bf {#1}}

We introduce the fragments of classical first-order logic (FO) that
are relevant for this paper. We generally admit equality and disallow
function symbols and constants. With \emph{FO$^2$}, we denote the
\emph{two-variable fragment of FO}, obtained by fixing two variables
$x$ and $y$ and disallowing the use of other
variables~\cite{Scott,DBLP:journals/mlq/Mortimer75}. In FO$^2$ and fragments thereof,
we generally admit  only predicates of arity one and two, which is
without loss of generality~\cite{DBLP:journals/bsl/GradelKV97}.  In the
\emph{guarded fragment of FO}, denoted \emph{GF}, quantification is
restricted to the pattern
$$
\forall \vec{y}(\alpha(\vec{x},\vec{y})\rightarrow \varphi(\vec{x},\vec{y}))
\quad
\exists \vec{y}(\alpha(\vec{x},\vec{y})\wedge \varphi(\vec{x},\vec{y}))
$$
where $\varphi(\vec{x},\vec{y})$ is a GF formula with free variables
among $\vec{x},\vec{y}$ and $\alpha(\vec{x},\vec{y})$ is an atomic
formula $R\xbf\ybf$ or an equality $x=y$ that in either case
contains all variables in $\vec{x},\vec{y}$~\cite{ANvB98,DBLP:journals/jsyml/Gradel99}. The
formula $\alpha$ is called the \emph{guard} of the quantifier. 
The $k$-variable fragment of GF, defined in the expected way, is
denoted \emph{GF$^{\,k}$}. Apart from the logics introduced so far, in
informal contexts we shall also mention several related description
logics. Exact definitions are omitted, we refer the reader to
\cite{BCMNP03}.

A \emph{signature} $\Sigma$ is a finite set of predicates.  We use
GF$(\Sigma)$ to denote the set of all GF-sentences that use only
predicates from $\Sigma$ (and possibly equality), and likewise for
GF$^2(\Sigma)$ and other fragments. We use $\mn{sig}(\varphi)$ to
denote the set of predicates that occur in the FO formula
$\varphi$. Note that we consider equality to be a logical symbol,
rather than a predicate, and it is thus never part of a signature. We write
$\vp_1 \models \vp_2$ if $\vp_2$ is a logical consequence of~$\vp_1$.
The next definition introduces the central notions studied in this paper.
\begin{definition}
\label{def:basicnotions}
  Let $F$ be a fragment of FO, $\vp_1, \vp_2$ $F$-sentences and
  $\Sigma$ a signature. Then
  \begin{enumerate}

  \item $\vp_1$ \emph{$\Sigma$-entails} $\vp_2$, written $\vp_1
    \models_\Sigma \vp_2$, if for all $F(\Sigma)$-sentences $\psi$,
    $\vp_2 \models \psi$ implies $\vp_1 \models \psi$;

  \item $\vp_1$ and $\vp_2$ are \emph{$\Sigma$-inseparable} if
    $\vp_1$ $\Sigma$-entails $\vp_2$ and vice versa;

  \item $\vp_1 \wedge \vp_2$ is a \emph{conservative extension} of
    $\vp_1$ if $\vp_1$ $\mn{sig}(\vp_1)$-entails $\vp_1 \wedge \vp_2$.

  \end{enumerate}
\end{definition}
Note that $\Sigma$-entailment could equivalently be defined as follows
when $F$ is closed under negation: $\vp_1$ \emph{$\Sigma$-entails}
$\vp_2$ if for all $F(\Sigma)$-sentences $\psi$, satisfiability of
$\vp_1 \wedge \psi$ implies satisfiability of $\vp_2 \wedge \psi$. If
$\vp_1$ does not $\Sigma$-entail $\vp_2$, there is thus an
$F(\Sigma)$-sentence $\psi$ such that $\vp_1 \wedge \psi$ is
satisfiable while $\vp_2 \wedge \psi$ is not. We refer to such $\psi$
as a \emph{witness sentence} for non-$\Sigma$-entailment.
\begin{example}\label{ex:ex3}
  (1) $\Sigma$-entailment is a weakening of logical consequence,
  that is, $\vp_1 \models \vp_2$ implies $\vp_1 \models_\Sigma \vp_2$
  for any $\Sigma$. The converse is true when $\mn{sig}(\vp_2)
  \subseteq \Sigma$. 

  \smallskip
  (2) Consider the GF$^2$ sentences $\vp_{1}= \forall x \exists y Rxy$ and
  $\vp_{2}=\forall x (\exists y (Rxy \wedge Ay) \wedge \exists y (Rxy \wedge \neg Ay))$ and let $\Sigma=\{R\}$.
  Then $\psi= \forall xy (Rxy \rightarrow x=y)$ is a witness for $\vp_{1}\not\models_{\Sigma}\vp_{2}$.
  If $\vp_{1}$ is replaced by $\vp_{1}'= \forall x \exists y (Rxy \wedge x\not=y)$ we obtain $\vp_{1}'\models_{\Sigma}\vp_{2}$
  since GF$^2$ cannot count the number of $R$-successors. 
\end{example}

It is important to note that different fragments $F$ of FO give rise
to different notions of $\Sigma$-entailment, $\Sigma$-inseparability
and conservative extensions. For example, if $\vp_1$ and $\vp_2$
belong to GF$^2$, then they also belong to GF and to FO$^2$, but it
might make a difference whether witness sentences range over all
GF$^2$-sentences, over all GF-sentences, or over all
FO$^2$-sentences. If we want to emphasize the fragment $F$ in which
witness sentences are formulated, we speak of $F(\Sigma)$-entailment
instead of $\Sigma$-entailment and write $\vp_1 \models_{F(\Sigma)}
\vp_2$, and likewise for $F(\Sigma)$-inseparability and
$F$-conservative extensions.
\begin{example}
Let $\vp_{1}'$, $\vp_{2}$, and $\Sigma=\{R\}$ be from
Example~\ref{ex:ex3} (2). Then $\vp_{1}'$ GF$^{2}(\Sigma)$-entails
$\vp_{2}$ but $\vp_{1}'$ does not FO$(\Sigma)$-entail $\vp_{2}$; a witness is given by 
$\forall xy_{1}y_{2}((Rxy_{1} \wedge Rxy_{2})\rightarrow y_{1}=y_{2})$.
\end{example}
Note that conservative extensions and $\Sigma$-inseparability reduce
in polynomial time to $\Sigma$-entailment (with two calls to
$\Sigma$-entailment required in the case of $\Sigma$-inseparability).
Moreover, conservative extensions reduce in polynomial time to
$\Sigma$-inseparability. We thus state our upper bounds in terms of
$\Sigma$-entailment and lower bounds in terms of conservative
extensions.

\smallskip

There is a natural variation of each of the three notions in
Definition~\ref{def:basicnotions} obtained by allowing to use
additional `helper predicates' in witness sentences. For a fragment
$F$ of FO, $F$-sentences $\vp_{1},\vp_{2}$, and a signature $\Sigma$,
we say that $\vp_{1}$ \emph{strongly $\Sigma$-entails} $\vp_{2}$ if
$\vp_{1}$ $\Sigma'$-entails $\vp_{2}$ for any $\Sigma'$ with
$\Sigma'\cap \sig(\vp_{2}) \subseteq \Sigma$. Strong
$\Sigma$-inseparability and strong conservative extensions are defined
accordingly. Strong $\Sigma$-entailment implies $\Sigma$-entailment,
but the converse may fail.
\begin{example}\label{ex:3}
GF$(\Sigma)$-entailment does not imply strong GF$(\Sigma)$-entailment.
Let $\vp_{1}$ state that the binary predicate $R$ is irreflexive and symmetric and let $\vp_{2}$ be the 
conjunction of $\vp_{1}$ and
$
\forall x (Ax \rightarrow \forall y (Rxy \rightarrow \neg Ay)) \wedge
\forall x (\neg Ax \rightarrow \forall y (Rxy \rightarrow Ay))$.
Thus, an $\{R\}$-structure satisfying $\vp_{1}$ can be extended to a model of $\vp_{2}$ if it contains
no $R$-cycles of odd length. Now observe that any satisfiable GF$(\{R\})$ sentence is satisfiable in a forest $\{R\}$-structure 
(see Section~\ref{sect:characterization} for a precise definition). Hence,
if a GF$(\{R\})$-sentence is satisfiable in an irreflexive and symmetric structure then it is satisfiable in a structure without
odd cycles and so $\vp_{1}$ GF$(\{R\})$-entails $\vp_{2}$. In contrast, for the fresh ternary predicate $Q$ and
$
\psi= \exists x_{1}x_{2}x_{3} (Qx_{1}x_{2}x_{3} \wedge Rx_{1}x_{2} \wedge Rx_{2}x_{3} \wedge Rx_{3}x_{1})
$
we have $\vp_{2}\models \neg \psi$ but $\vp_{1}\not\models \neg \psi$ and so $\psi$ witnesses that
$\vp_{1}$ does not GF$(\{R,Q\})$-entail $\vp_{2}$.
\end{example}
The example above is inspired by proofs that GF does not enjoy Craig
interpolation~\cite{DBLP:journals/sLogica/HooglandM02,DBLP:journals/tcs/DAgostinoL15}.
This is not accidental, as we explain next. Recall that a fragment $F$
of FO \emph{has Craig interpolation} if for all $F$-sentences
$\psi_{1},\psi_{2}$ with $\psi_{1}\models \psi_{2}$ there exists an
$F$-sentence $\psi$ (called an \emph{$F$-interpolant for
  $\psi_{1},\psi_{2}$}) such that $\psi_{1}\models \psi\models
\psi_{2}$ and $\sig(\psi) \subseteq \sig(\psi_{1})\cap
\sig(\psi_{2})$. $F$ \emph{has uniform interpolation} if one can
always choose an $F$-interpolant that does not depend on $\psi_{2}$,
but only on $\psi_{1}$ and $\sig(\psi_{1})\cap \sig(\psi_{2})$. Thus,
given $\psi_{1}$, $\psi$ and $\Sigma$ with $\psi_{1}\models \psi$ and
$\sig(\psi)\subseteq \Sigma$, then $\psi$ is a \emph{uniform
  $F(\Sigma)$-interpolant of $\psi_{1}$} iff $\psi$ strongly
$F(\Sigma)$-entails $\psi_{1}$. Both Craig interpolation and uniform
interpolation have been investigated extensively, for example for
intuitionistic logic \cite{Pitts}, modal
logics~\cite{Visser,DBLP:journals/jsyml/DAgostinoH00,DBLP:conf/calco/MartiSV15},
guarded fragments \cite{DBLP:journals/tcs/DAgostinoL15}, and
description logics \cite{DBLP:conf/ijcai/LutzW11}. The following
observation summarizes the connection between (strong)
$\Sigma$-entailment and interpolation.
\begin{theorem}\label{thm:interpol}
Let $F$ be a fragment of FO that enjoys Craig interpolation. Then $F(\Sigma)$-entailment implies strong $F(\Sigma)$-entailment. In particular,
if $\vp_{2}\models \vp_{1}$ and $\sig(\vp_{1})\subseteq \Sigma$, then $\vp_{1}$ is a uniform $F(\Sigma)$-interpolant
of $\vp_{2}$ iff $\vp_{1}$ $F(\Sigma)$-entails $\vp_{2}$.
\end{theorem}
\begin{proof}
Assume $\vp_{1}$ does not strongly $F(\Sigma)$-entail $\vp_{2}$.
Then there is an $F$-sentence $\psi$ with $\sig(\psi)\cap \sig(\vp_{2})\subseteq \Sigma$ such that
$\vp_{2}\models \psi$ and $\vp_{1} \wedge \neg\psi$ is satisfiable. Let $\chi$ be an interpolant for $\vp_{2}$ and
$\psi$ in $F$. Then $\neg\chi$ witnesses that $\vp_{1}$ does not $F(\Sigma)$-entail $\vp_{2}$.
\end{proof}
The \emph{uniform interpolant recognition problem for $F$} is the
problem to decide whether a sentence $\psi$ is a uniform
$F(\Sigma)$-interpolant of a sentence $\psi'$. It follows from
Theorem~\ref{thm:interpol} that in any fragment $F$ of FO that enjoys
Craig interpolation, this problem
reduces in polynomial time to $\Sigma$-inseparability in $F$ and that,
conversely, conservative extension in $F$ reduces in polynomial time
to the uniform interpolant recognition problem in $F$.
Neither GF nor FO$^2$ nor description logics with role inclusions
enjoy Craig
interpolation~\cite{DBLP:journals/sLogica/HooglandM02,SComer,DBLP:series/lncs/KonevLWW09},
but GF$^2$ does 
\cite{DBLP:journals/sLogica/HooglandM02}.
Thus, our decidability and  complexity results for $\Sigma$-entailment in GF$^2$ also apply to strong $\Sigma$-entailment and the
uniform interpolant recognition problem.
\section{Undecidability}
\label{sect:undec}

We prove that conservative extensions are undecidable in GF$^3$ and in
FO$^2$, and consequently so are $\Sigma$-entailment and
$\Sigma$-inseparability (as well as strong $\Sigma$-entailment and the
uniform interpolant recognition problem).  These results hold already
without equality and in fact apply to all fragments of FO that contain
at least one of GF$^3$ and FO$^2$ such as the guarded negation
fragment~\cite{DBLP:journals/jacm/BaranyCS15} and the two-variable
fragment with counting quantifiers \cite{GOR97-C2}. 

We start with the
case of GF$^3$, using a reduction from the halting problem of
two-register machines. 
  A (deterministic) \emph{two-register machine
    (2RM)}\index{Two-register machine}\index{2RM} is a pair
  $M=(Q,P)$ with $Q = q_0,\dots,q_{\ell}$ a set of \emph{states}
  and $P = I_0,\dots,I_{\ell-1}$ a sequence of \emph{instructions}.
  By definition, $q_0$ is the \emph{initial state}, and $q_\ell$ the
  \emph{halting state}.  For all $i < \ell$,
  \begin{itemize}

  \item either $I_i=+(p,q_j)$ is an \emph{incrementation instruction}
    with $p \in \{0,1\}$ a register and $q_j$ the subsequent
    state;

  \item or $I_i=-(p,q_j,q_k)$ is a \emph{decrementation instruction} with
    $p \in \{0,1\}$ a register, $q_j$ the subsequent state if
    register~$p$ contains~0, and $q_k$ the subsequent state
    otherwise.

  \end{itemize}
  A \emph{configuration} of $M$ is a triple $(q,m,n)$, with $q$ the
  current state and $m,n \in \mathbbm{N}$ the register contents.  We
  write $(q_i,n_1,n_2) \Rightarrow_M (q_j,m_1,m_2)$ if one of the
  following holds:
  \begin{itemize}

  \item $I_i = +(p,q_j)$, $m_p = n_p +1$, and $m_{1-p} =
    n_{1-p}$;

  \item $I_i = -(p,q_j,q_k)$, $n_p=m_p=0$, and $m_{1-p} =
    n_{1-p}$;

  \item $I_i = -(p,q_k,q_j)$, $n_p > 0$, $m_p = n_p -1$, and
    $m_{1-p} = n_{1-p}$.

  \end{itemize}
  The \emph{computation} of $M$ on input $(n,m) \in \mathbbm{N}^2$ is
  the unique longest configuration sequence $(p_0,n_0,m_0)
  \Rightarrow_M (p_1,n_1,m_1) \Rightarrow_M \cdots$ such that $p_0 =
  q_0$, $n_0 = n$, and $m_0 = m$.
The halting problem for 2RMs is to decide, given a 2RM $M$, whether
its computation on input $(0,0)$ is finite (which implies that its
last state is $q_\ell$).

We show how to convert a given 2RM $M$ into GF$^{3}$-sentences
$\vp_{1}$ and $\vp_{2}$ such that $M$ halts on input $(0,0)$ iff
$\vp_{1}\wedge \vp_{2}$ is not a conservative extension of $\vp_{1}$.
Let $M=(Q,P)$ with $Q = q_0,\dots,q_{\ell}$ and
$P = I_0,\dots,I_{\ell-1}$. We assume w.l.o.g.\ that $\ell \geq 1$ and
that if $I_i=-(p,q_j,q_k)$, then $q_j \neq q_k$. In $\vp_{1}$,
we use the following set $\Sigma$ of predicates:
\begin{itemize}
\item a binary predicate $N$ connecting a configuration to its successor configuration;
\item binary predicates $R_{1}$ and $R_{2}$ that represent the register
contents via the length of paths;
\item unary predicates $q_{0},\ldots,q_{\ell}$ representing the states of $M$;
\item a unary predicate $S$ denoting points where a computation starts.
\end{itemize}
We define $\vp_1$ to be the conjunction of several
GF$^2$-sentences. First, we say that there is a point where the
computation starts:\footnote{The formulas that are not syntactically 
guarded can easily be rewritten into
such
formulas.}
$$
\exists x  Sx 
\wedge \forall x  (Sx \rightarrow (q_0x \wedge \neg \exists y \, R_0xy  \wedge \neg
\exists y \, R_1xy))
$$
And second, we add that whenever $M$ is not in the final state,
there is a next configuration. For $0 \leq i < \ell$:
\vspace*{-4mm}
  $$
  \begin{array}{rl}
    \forall x (q_{i}x \rightarrow \exists y (Nxy \wedge q_{j}y)) & \text{ if } I_i=+(p,q_j) \\[1mm]
    \forall x ((q_{i}x \wedge \neg \exists y R_{p}xy) \rightarrow\exists y (Nxy \wedge q_{j}y)) & \text{ if } I_i=-(p,q_j,q_k) \\[1mm]
     \forall x ((q_{i}x \wedge \exists y R_{p}xy)  \rightarrow \exists y (Nxy \wedge q_{k}y)) & \text{ if } I_i=-(p,q_j,q_k)
  \end{array}
  $$
  The second sentence $\varphi_2$ is constructed so as to express that
  either $M$ does not halt or the representation of the computation of
  $M$ contains a defect, using the following additional predicates:
\begin{itemize}
\item a unary predicate $P$ used to represent that $M$ does not halt;
\item binary predicates $D^{+}_{p},D_{p}^{-},D_{p}^{=}$ used to
  describe defects in incrementing, decrementing, and keeping register $p\in \{0,1\}$;
\item ternary predicates $H_{1}^{+},H_{2}^{+},H_{1}^{-},H_{2}^{-},H_{1}^{=},H_{2}^{=}$ used as guards for existential quantifiers.
\end{itemize}
In fact, $\vp_{2}$ is the disjunction of two sentences. The first
sentence says that the computation does not terminate:
$$
   \exists x \, (Sx \wedge Px) 
   \wedge  \forall x \, (Px \rightarrow \exists y \, (Nxy \wedge Py))
$$
while the second says that registers are not updated
properly:
  $$
  \begin{array}{l}
  \exists x \exists y \, \big (Nxy \wedge 
\displaystyle    \big(\bigvee_{I_i=+(p,q_j)} (q_ix \wedge q_jy \wedge
    (D^+_pxy \vee D^=_{1-p}xy)) \\[6mm]
    \hspace*{29mm}
\displaystyle    \vee\,
    \bigvee_{I_i=-(p,q_j,q_k)} (q_ix \wedge q_ky \wedge (D^-_pxy \vee
    D^=_{1-p}xy))
\\[6mm]
    \hspace*{29mm}
\displaystyle    \vee\,
    \bigvee_{I_i=-(p,q_j,q_k)} (q_ix \wedge q_jy \wedge (D^=_pxy \vee
    D^=_{1-p}xy))\big)\big)
 \\[6mm]
  \wedge \, \forall x \forall y \, (D^+_pxy \rightarrow 
(\neg \exists z \, R_pyz \vee (\neg \exists z \, R_pxz
    \wedge \exists z \, (R_pyz \wedge \exists x R_pzx))\\[2mm]
   \hspace*{29mm}\vee \, 
\exists z (H_{1}^{+}xyz \wedge R_{p}xz \wedge \exists x (H_{2}^{+}xzy \wedge R_pyx \wedge D^+_{p}zx)).
  \end{array}
  $$
  In this second sentence, additional conjuncts that implement the
  desired behaviour of $D^=_p$ and $D^-_p$ are also needed; they are
  constructed analogously to the last three lines above (but using guards
  $H^{-}_j$ and $H^{=}_j$), details are omitted.
The following is proved in the appendix of this paper. 
\begin{lemma}\label{lem:gfundec1plus}
~\\[-4mm]
  \begin{enumerate}

  \item
    If $M$ halts, then $\vp_{1} \wedge \vp_{2}$ is not a
    $\text{GF}^{\,2}$-conservative extension of $\vp_{1}$.

  \item If there exists a $\Sigma$-structure that satisfies $\vp_{1}$
    and cannot be extended to a model of $\vp_{2}$ (by interpreting
    the predicates in $\mn{sig}(\vp_2) \setminus \mn{sig}(\vp_1)$),
    then $M$ halts.
  \end{enumerate}
\end{lemma}
In the proof of Point~1, the sentence that witnesses
non-conservativity describes a halting computation of $M$, up to
global GF$^2(\Sigma)$-bisimulations. This can be done using only two
variables. The following result is an immediate consequence of
Lemma~\ref{lem:gfundec1plus}.
\begin{theorem}\label{mainth0}
In any fragment of FO that extends the three-variable guarded fragment
GF$^{\,3}$, the following problems
  are undecidable: conservative extensions, $\Sigma$-inseparability,
  $\Sigma$-entailment, and strong
  $\Sigma$-entailment.
\end{theorem}
Since Point~1 of Lemma~\ref{lem:gfundec1plus} ensures
GF$^2$-witnesses, Theorem~\ref{mainth0} can actually be 
strengthened to say that GF$^2$-conservative extensions of
GF$^3$-sentences are undecidable.

\medskip

Our result for FO$^2$ is proved by a reduction of a tiling problem
that asks for the tiling of a rectangle (of any size) such that the
borders are tiled with certain distinguished tiles. Because of space
limitations, we defer details to the appendix and
state only the obtained result.
\begin{theorem}\label{mainth0fo2}
  In any fragment of FO that extends FO$^2$, the following problems
  are undecidable: conservative extensions, $\Sigma$-inseparability,
  $\Sigma$-entailment, and strong
  $\Sigma$-entailment.
\end{theorem}
It is interesting to note that the proof of Theorem~\ref{mainth0fo2} also
shows that FO$^2$-conservative extensions of
$\mathcal{ALC}$-TBoxes are undecidable while it follows from our
results below that GF$^2$-conservative extensions of $\mathcal{ALC}$-TBoxes are decidable.

\section{Characterizations}
\label{sect:characterization}

The undecidability results established in the previous section show
that neither the restriction to two variables nor guardedness alone
are sufficient for decidability of conservative extensions and related
problems. In the remainder of the paper, we show that adopting both
restrictions simultaneously results in decidability of
$\Sigma$-entailment (and thus also of conservative extensions and of
inseparability). We proceed by first establishing a suitable
model-theoretic characterization and then use it as the
foundation for a decision procedure based on tree automata. We in fact
establish two versions of the characterization, the second one
building on the first one.

We start with some preliminaries. An \emph{atomic 1-type for $\Sigma$}
is a maximal satisfiable set $\tau$ of atomic
GF$^{2}(\Sigma)$-formulas and their negations that use the variable
$x$, only. We use $\text{at}_{\mathfrak{A}}^{\Sigma}(a)$ to denote the
atomic 1-type for $\Sigma$ realized by the element $a$ in the
structure $\mathfrak{A}$. An \emph{atomic 2-type for $\Sigma$} is a
maximal satisfiable set $\tau$ of atomic GF$^{2}(\Sigma)$-formulas and
their negations that use the variables $x$ and $y$, only, and contains
$\neg(x=y)$. We say that $\tau$ is \emph{guarded} if it contains an
atom of the form $Rxy$ or $Ryx$, $R$ a predicate symbol.  We use
$\text{at}_{\mathfrak{A}}^{\Sigma}(a,b)$ to denote the atomic 2-type
for $\Sigma$ realized by the elements $a,b$ in the structure
$\mathfrak{A}$. A relation ${\sim} \subseteq A \times B$ is a
\emph{GF$^{2}(\Sigma)$-bisimulation between \Amf and \Bmf} if the
following conditions hold whenever $a \sim b$:
\begin{enumerate}

\item $\text{at}_{\Amf}^{\Sigma}(a)= \text{at}_{\Bmf}^{\Sigma}(b)$;

\item for every $a'\not=a$ such that $\text{at}_{\Amf}^{\Sigma}(a,a')$
  is guarded, there is a $b'\not=b$ such that
$\text{at}_{\Amf}^{\Sigma}(a,a')= \text{at}_{\Bmf}^{\Sigma}(b,b')$ and
$a'\sim b'$ (forth condition);

\item for every $b'\not=b$ such that $\text{at}_{\Bmf}^{\Sigma}(b,b')$
  is guarded, there is an $a'\not=a$ such that
$\text{at}_{\Amf}^{\Sigma}(a,a')= \text{at}_{\Bmf}^{\Sigma}(b,b')$ and
$a'\sim b'$  (back condition).
\end{enumerate}
We write $(\Amf,a) \sim_\Sigma (\Bmf,b)$ and say that $(\Amf,a)$ and
$(\Bmf,b)$ are \emph{GF$^{\,2}(\Sigma)$-bisimilar} if there is a
GF$^{2}(\Sigma)$-bisimulation $\sim$ between \Amf and \Bmf with $a
\sim b$. If the domain and range of $\sim$ coincide with $A$ and $B$,
respectively, then $\sim$ is called a \emph{global
  GF$^{\,2}(\Sigma)$-bisimulation}.

We next introduce a bounded version of bisimulations.  For $k \geq 0$,
we write $(\Amf,a) \sim^k_\Sigma (\Bmf,b)$ and say that $(\Amf,a)$ and
$(\Bmf,b)$ are \emph{$k$-GF$^{\,2}(\Sigma)$-bisimilar} if there is a
${\sim} \subseteq A \times B$ such that the first condition for
bisimulations holds and the back and forth conditions can be iterated
up to $k$ times starting from $a$ and $b$; a formal definition is in
the appendix. It is straightforward to show the
following link between $k$-GF$^2$-bisimilarity and GF$^2$-sentences of
guarded quantifier depth~$k$ (defined in the obvious way).
\begin{lemma}\label{lem:k}
Let $\Amf$ and $\Bmf$ be structures, $\Sigma$ a signature, and $k\geq 0$. Then the following conditions
are equivalent:
\begin{enumerate}
\item for all $a\in A$ there exists $b\in B$ with $(\Amf,a) \sim^k_\Sigma (\Bmf,b)$ and vice versa;
\item $\Amf$ and $\Bmf$ satisfy the same GF$^{\,2}(\Sigma)$-sentences of guarded quantifier depth at most $k$.
\end{enumerate}
\end{lemma} 
The corresponding lemma for GF$^{2}(\Sigma)$-sentences of unbounded
guarded quantifier depth and GF$^{2}(\Sigma)$-bisimulations holds if
$\Amf$ and $\Bmf$ satisfy certain saturation conditions (for example,
if $\Amf$ and $\Bmf$ are $\omega$-saturated). It can then be proved
that an FO-sentence $\vp$ is equivalent to a GF$^2$ sentence iff its
models are preserved under global GF$^{2}(\sig(\vp))$-bisimulations
\cite{DBLP:books/daglib/p/Gradel014,GorankoOtto}. In modal and
description logic, global $\Sigma$-bisimulations can often be used to
characterize $\Sigma$-entailment in the following natural way
\cite{DBLP:conf/ijcai/LutzW11}: $\vp_{1}$ $\Sigma$-entails $\vp_{2}$
iff every for every (tree) model $\Amf$ of $\vp_{1}$, there exists a
(tree) model $\Bmf$ of $\vp_{2}$ that is globally $\Sigma$-bisimilar
to $\Amf$. Such a characterization enables decision procedures based
on tree automata, but does not hold for GF$^2$.
\begin{example}
  Let $\vp_{1}= \forall x \exists y Rxy$ and let $\vp_{2} = \vp_{1}
  \wedge \exists x Bx \wedge \forall x( Bx \rightarrow \exists y (Ryx
  \wedge By))$. Let $\Amf$ be the model of $\vp_{1}$ that consists of
  an infinite $R$-path with an initial element. Then there is no
  model of $\vp_{2}$ that is globally GF$^{2}(\{R\})$-bisimilar to
  $\Amf$ since any such model has to contain an infinite $R$-path
  with no initial element. Yet, $\vp_{2}$ is a
  conservative extension of $\vp_{1}$ which can be proved using
  Theorem~\ref{thm:char1} below.
\end{example}

We give our first characterization theorem that uses unbounded
bisimulations in one direction and bounded bisimulations in the
other.
\begin{theorem}
\label{thm:char1} 
Let $\varphi_1,\varphi_2$ be $GF_2$-sentences and $\Sigma$ a
signature. Then $\vp_1 \models_\Sigma \vp_2$ iff for every model \Amf of $\vp_1$
of finite outdegree, there is a model \Bmf of $\vp_2$ such that
  \begin{enumerate}

  \item for every $a\in A$ there is a $b
    \in B$ such that $(\Amf,a) \sim_\Sigma (\Bmf,b)$

  \item for every $b\in B$ and every $k
    \geq 0$, there is an $a \in A$ such that $(\Amf,a) \sim^k_\Sigma (\Bmf,b)$.

  \end{enumerate}
\end{theorem}
The direction ($\Leftarrow$) follows from Lemma~\ref{lem:k} and
($\Rightarrow$) can be proved using compactness and $\omega$-saturated
structures. Because of the use of $k$-bounded bisimulations (for
unbounded $k$), it is not clear how to use Theorem~\ref{thm:char1} to
find a decision procedure based on tree automata. In the following, we
formulate a more `operational' but also more technical
characterization that no longer mentions bounded bisimulations. It
additionally refers to forest models \Amf of $\vp_1$ (of finite
outdegree) instead of unrestricted models, but we remark that
Theorem~\ref{thm:char1} also remains true under this modification.

A structure \Amf is a \emph{forest} if its Gaifman graph is a
forest. Thus, a forest admits cycles of length one and two, but not of
any higher length.  A ($\Sigma$-)\emph{tree} in a forest structure
\Amf is a maximal ($\Sigma$)-connected substructure of~\Amf. When
working with forest structures \Amf, we will typically view them as
directed forests rather than as undirected ones. This can be done by
choosing a root for each tree in the Gaifman graph of \Amf, thus
giving rise to notions such as successor, descendant, etc. Which node
is chosen as the root will always be irrelevant. Note that the
direction of binary relations does not need to reflect the successor
relation.  When speaking of a \emph{path} in a forest structure \Amf,
we mean a path in the directed sense; when speaking of a
\emph{subtree}, we mean a tree that is obtained by choosing
a root $a$ and restricting the structure to $a$ and its descendants.
We say that \Amf is \emph{regular} if it has only finitely many
subtrees, up to isomorphism.

To see how we can get rid of bounded bisimulations, reconsider
Theorem~\ref{thm:char1}. The characterization is still correct if we
pull out the quantification over $k$ in Point~2 so that the theorem
reads `...iff for every model \Amf of $\vp_1$ of finite outdegree and
every $k \geq 0$, there is...'. In fact, this modified version of
Theorem~\ref{thm:char1} is even closer to the definition of
$\Sigma$-entailment. It also suggests that we add a marking
$A_\bot \subseteq A$ of elements in $\Amf$, representing `break-off
points' for bisimulations, and then replace $k$-bisimulations with
bisimulations that stop whenever they have encountered the
\emph{second} marked element on the same path---in this way, the
distance between marked elements (roughly) corresponds to the bound
$k$ in $k$-bisimulations. However, we would need a marking $A_\bot$,
for any $k \geq 0$, such that there are infinitely many markers on any
infinite path and the distance between any two markers in a tree is at
least~$k$. It is easy to see that such a marking may not exist, for
example when $k=3$ and \Amf is the infinite full binary tree. We solve
this problem as follows. First, we only demand that the distance
between any two markers \emph{on the same path} is at least $k$.
And second, we use the markers only when following bisimulations
upwards in a tree while downwards, we use unbounded
bisimulations. This does not compromise correctness of the
characterization. 

We next introduce a version of bisimulations that implement the ideas
just explained.
Let $\Amf$ and $\Bmf$ be forest models, $\Sigma$ a signature, and
$A_{\bot}\subseteq A$. Two relations ${\sim}^{A_\bot,0}_{\Sigma},{\sim}^{A_\bot,1}_{\Sigma}\subseteq A
\times B$ form an \emph{$A_{\bot}$-delimited
  GF$^{\,2}(\Sigma)$-bisimulation} between $\Amf$ and $\Bmf$ if the
following conditions are satisfied:
\begin{enumerate}

\item if $(\Amf,a) \sim^{A_\bot,0}_\Sigma (\Bmf,b)$, then
  $\text{at}_{\Amf}^{\Sigma}(a)=\text{at}_{\Bmf}^{\Sigma}(b)$ and
  \begin{enumerate}

  \item for every $a' \neq a$ with $\text{at}_{\Amf}^{\Sigma}(a,a')$
    guarded, there is a $b'\not=b$ such that $(\Amf,a')
    \sim^{A_\bot, i}_\Sigma (\Bmf,b')$ where $i=1$ if $a'$ is the
    predecessor of $a$ and $a' \in
    A_\bot$, and $i=0$ otherwise;

  \item for every $b' \neq b$ with $\text{at}_{\Bmf}^{\Sigma}(b,b')$
    guarded, there is an $a'\not=a$ such that $(\Amf,a')
    \sim^{A_\bot, i}_\Sigma (\Bmf,b')$ where $i=1$ if $a'$ is the
    predecessor of $a$ and $a' \in
    A_\bot$, and $i=0$ otherwise;
    
  \end{enumerate}

\item if $(\Amf,a) \sim^{A_\bot,1}_\Sigma (\Bmf,b)$ and the
  predecessor of $a$ in \Amf is not in $A_\bot$, then
  $\text{at}_{\Amf}^{\Sigma}(a)=\text{at}_{\Bmf}^{\Sigma}(b)$ and
  \begin{enumerate}

  \item for every $a' \neq a$ with $\text{at}_{\Amf}^{\Sigma}(a,a')$
    guarded, there is a $b'\not=b$ such that $(\Amf,a') \sim^{A_\bot,
    i}_\Sigma (\Bmf,b')$ where $i=0$ if $a$ is the predecessor of $a'$
    and $a \in A_\bot$, and $i=1$ otherwise;

  \item for every $b' \neq b$ with $\text{at}_{\Bmf}^{\Sigma}(b,b')$
    guarded, there is an $a'\not=a$ such that $(\Amf,a') \sim^{A_\bot,
    i}_\Sigma (\Bmf,b')$ where $i=0$ if $a$ is the predecessor of $a'$
    and $a \in A_\bot$, and $i=1$ otherwise.
    
  \end{enumerate}

\end{enumerate}

Then $(\Amf,a)$ and $(\Bmf,b)$ are \emph{$A_{\bot}$-delimited
  GF$^{\,2}(\Sigma)$-bisimilar}, in symbols
$(\Amf,a)\sim_{\Sigma}^{A_{\bot}}(\Bmf,b)$, if there exists an
$A_{\bot}$-delimited GF$^{2}(\Sigma)$-bisimulation
${\sim}^{A_\bot,0}_{\Sigma},{\sim}^{A_\bot,1}_{\Sigma}$ between
$\Amf$ and $\Bmf$ such that $(\Amf,a) \sim^{A_\bot,0}_{\Sigma}
(\Bmf,b)$. 

Let $\varphi$ be a GF$^2$-sentence.  We use $\mn{cl}(\varphi)$ to
denote the set of all subformulas of $\varphi$ closed under single
negation and renaming of free variables (using only the available
variables $x$ and $y$). A {\em $1$-type for $\varphi$} is a subset $t
\subseteq \mn{cl}(\varphi)$ that contains only formulas of the form
$\psi(x)$ and such that $\varphi\wedge\exists x\, \bigwedge t(x)$ is
satisfiable. For a model \Amf of $\vp$ and $a \in A$, we use
$\mn{tp}_\Amf(a)$ to denote the $1$-type $\{ \psi(x) \in \mn{cl}(\vp)
\mid \Amf \models \psi(a) \}$, assuming that $\vp$ is understood from
the context. We say that the $1$-type $t$ is \emph{realized} in \Amf
if there is an $a \in A$ with $\mn{tp}_\Amf(a)=t$. We are now ready to
formulate our final characterizations.
\begin{theorem} \label{thm:charsimp}
  Let $\varphi_1,\varphi_2$ be GF$^{\,2}$-sentences and $\Sigma$ a
  signature. Then $\vp_1 \models_\Sigma \vp_2$ iff for every regular
  forest model \Amf of $\vp_1$ that 
  has finite outdegree and for every set $A_{\bot}\subseteq A$
  with $A_{\bot}\cap \rho$ infinite for any infinite $\Sigma$-path $\rho$ in $\Amf$,
  there is a model \Bmf of $\vp_2$ such that
  \begin{enumerate}

  \item for every $a \in A$, there is a $b \in B$ such that $(\Amf,a) \sim_\Sigma (\Bmf,b)$;

  \item for every 1-type $t$ for $\vp_2$ that is realized in \Bmf,
    there are $a \in A$ and $b \in B$ such that $\mn{tp}_\Bmf(b)=t$
    and $(\Amf,a)\sim_{\Sigma}^{A_{\bot}}(\Bmf,b)$.

\end{enumerate} 

\end{theorem}
Regularity and finite outdegree are used in the proof of
Theorem~\ref{thm:charsimp} given in the appendix,
but it follows from the automata constructions below that the theorem
is still correct when these qualifications are dropped.

\section{Decidability and Complexity}
\label{sect:upper}

We show that $\Sigma$-entailment in GF$^2$ is decidable and
2\ExpTime-complete, and thus so are conservative extensions and
$\Sigma$-inseparability. The upper bound is based on
Theorem~\ref{thm:charsimp} and uses alternating parity automata on
infinite trees.  Since Theorem~\ref{thm:charsimp} does not provide us
with an obvious upper bound on the outdegree of the involved tree
models, we use alternating tree automata which can deal
with trees of any finite outdegree, similar to the ones introduced by
Wilke \cite{Wilke}, but with the capability to move both downwards
and upwards in the tree. 

A \emph{tree} is a non-empty (and potentially infinite) set
of words $T \subseteq (\Nbbm \setminus 0)^*$ closed under prefixes.
We generally assume that trees are finitely branching, that is, for
every $w \in T$, the set $\{ i \mid w \cdot i \in T \}$ is finite.
For any $w \in (\Nbbm \setminus 0)^*$, as a convention we set $w
\cdot 0 := w$. If $w=n_0n_1 \cdots n_k$, 
we additionally
set $w \cdot -1 := n_0 \cdots n_{k-1}$.  For an alphabet $\Theta$, a
\emph{$\Theta$-labeled tree} is a pair $(T,L)$ with $T$ a tree and
$L:T \rightarrow \Theta$ a node labeling function.

A {\em two-way alternating tree automata (2ATA)} is a tuple
$\Amc = (Q,\Theta,q_0,\delta,\Omega)$ where $Q$ is a finite set of
{\em states}, $\Theta$ is the {\em input alphabet}, $q_0\in Q$ is the
{\em initial state}, $\delta$ is a {\em transition function} as
specified below, and $\Omega:Q\to \mathbb{N}$ is a {\em priority
  function}, which assigns a priority to each state.
The transition function maps a state $q$ and some input letter
$\theta\in \Theta$ to a {\em transition condition $\delta(q,\theta)$}
which is a positive Boolean formula over the truth constants
$\mn{true}$ and $\mn{false}$ and transitions of the form $q$,
$\langle-\rangle q$, $[-] q$, $\Diamond q$, $\Box q$ where $q \in
Q$.
The automaton runs on $\Theta$-labeled trees.  Informally, the
transition $q$ expresses that a copy of the automaton is sent to the
current node in state $q$, $\langle - \rangle q$ means that a copy is
sent in state $q$ to the predecessor node, which is then required to
exist, $[-] q$ means the same except that the predecessor node is not
required to exist, $\Diamond q$ means that a copy is sent in state $q$
to some successor, and $\Box q$ that a copy is sent in state $q$ to
all successors. The semantics is defined in terms of runs in the usual
way, we refer to the appendix for details.  We use
$L(\Amc)$ to denote the set of all $\Theta$-labeled trees accepted by
\Amc.  It is standard to verify that 2ATAs are closed under
complementation and intersection. We show in the appendix that the
emptiness problem for 2ATAs can be solved in time exponential in the
number of states.

We aim to show that given two GF$^2$-sentences $\vp_1$ and $\vp_2$ and
a signature $\Sigma$, one can construct a 2ATA \Amc such that
$L(\Amc)=\emptyset$ iff $\vp_1 \models_{\text{GF}^2(\Sigma)}
\vp_2$. The number of states of the 2ATA \Amc is polynomial in
the size of $\vp_1$ and exponential in the size of $\vp_2$, which
yields the desired 2\ExpTime upper bounds.

Let $\varphi_1$, $\varphi_2$, and $\Sigma$ be given.  Since the logics
we are concerned with have Craig interpolation, we can assume
w.l.o.g.\ that $\Sigma \subseteq \mn{sig}(\vp_1)$.  With $\Theta$, we
denote the set of all pairs $(\tau,M)$ where $\tau$ is an atomic
2-type for $\mn{sig}(\varphi_1)$ and $M \in \{0,1\}$.  For $p=(\tau,M)
\in \Theta$, we use $p^1$ to denote $\tau$ and $p^2$ to denote $M$.  A
$\Theta$-labeled tree $(T,L)$ represents a forest structure
$\Amf_{(T,L)}$ with universe $A_{(T,L)}=T$ and where $w \in
A^{\Amf_{(T,L)}}$ if $A(y) \in L(w)$ and $(w,w') \in R^{\Amf_{(T,L)}}$
if one of the following conditions is satisfied: (1)~$w=w'$ and $Ryy
\in L(w)^1$; (2)~$w'$ is a successor of $w$ and $Rxy \in L(w')^1$;
(3)~$w$ is a successor of $w'$ and $Ryx \in L(w)^1$.  Thus, the atoms
in a node label that involve only the variable $y$ describe the
current node, the atoms that involve both variables $x$ and $y$
describe the connection between the predecessor and the current node,
and the atoms that involve only the variable $x$ are ignored. The
$M$-components of node labels are used to represent a set of markers
$A_\bot = \{ w \in A_{(T,L)} \mid L(w)^2=1\}$.  It is easy to see
that, conversely, for every tree structure \Amf over $\Sigma$, there
is a $\Theta$-labeled tree $(T,L)$ such that $\Amf_{(T,L)}=\Amf$.

To obtain the desired 2ATA \Amc, we construct three 2ATAs $\Amc_1,
\Amc_2, \Amc_3$ and then define \Amc so that it accepts $L(\Amc_1)
\cap \overline{L(\Amc_2)} \cap L(\Amc_3)$. The 2ATA $\Amc_3$ only
makes sure that the set $A_\bot \subseteq A_{(T,L)}$ is such that for
any infinite $\Sigma$-path $\rho$, $A_\bot \cap \rho$ is infinite (as
required by Theorem~\ref{thm:charsimp}), we omit details. We construct
$\Amc_1$ so that it accepts a $\Theta$-labeled tree $(T,L)$ iff
$\Amf_{(T,L)}$ is a model of $\vp_1$. The details of the construction,
which is fairly standard, can be found
in the appendix. The number of states of $\Amc_1$ is polynomial in the
size of $\vp_1$ and independent of $\vp_2$. The most interesting
automaton is $\Amc_2$.
\begin{lemma}
\label{lem:bisi-automata}
  There is a 2ATA $\Amc_2$ that accepts a $\Theta$-labeled tree
  $(T,L)$ iff there is a model~\Bmf of $\vp_2$ s.t.\
  Conditions~1 and~2 from Theorem~\ref{thm:charsimp} are satisfied when
  $\Amf$ is replaced with $\Amf_{(T,L)}$.
\end{lemma}
The general idea of the construction of $\Amc_2$
is to check the existence of the desired model \Bmf of $\vp_2$ by
verifying that there is a set of 1-types for $\vp_2$ from which \Bmf
can be assembled, represented via the states that occur in a
successful run. 
Before we can give details, we introduce some preliminaries.  

A {\em
  $0$-type $s$ for $\varphi_2$} is a maximal set of sentences
$\psi()\in\mn{cl}(\varphi_2)$ such that $\varphi_2\wedge s$ is
satisfiable. A {\em $2$-type $\lambda$ for $\varphi_2$} is a maximal
set of formulas $\psi(x,y)\in\mn{cl}(\varphi_2)$ that contains
$\neg(x=y)$ and such that $\varphi_2\wedge \exists xy\, \lambda(x,y)$
is satisfiable.  If a 2-type $\lambda$ contains the atom $Rxy$ or
$Ryx$ for at least one binary predicate $R$, then it is
\emph{guarded}. If additionally $R \in \Sigma$, then it is
\emph{$\Sigma$-guarded}.  Note that each $1$-type contains a (unique)
$0$-type and each $2$-type contains two (unique) $1$-types. Formally,
we use $\lambda_x$ to denote the 1-type obtained by restricting the
2-type $\lambda$ to the formulas that do not use the variable $y$, and
likewise for $\lambda_y$ and the variable $x$.  We use $\mn{TP}_n$ to
denote the set of $n$-types for $\varphi_2$, $n\in\{0,1,2\}$. For $t
\in \mn{TP}_1$ and a $\lambda \in \mn{TP}_2$, we say that $\lambda$ is
\emph{compatible with $t$} and write $t \approx \lambda$ if the
sentence $\varphi_2\wedge \exists xy (t(x) \wedge \lambda(x,y))$ is
satisfiable; for $t \in \mn{TP}_1$ and $T\subseteq\mn{TP}_2$ a set of
guarded 2-types, we say that $T$ is a \emph{neighborhood for $t$} and
write $t \approx T$ if the sentence
\begin{align*}
  \varphi_2\wedge \exists x\big(t(x)\wedge
  \bigwedge_{\lambda\in
  T}\exists y\, \lambda(x,y) \wedge \forall
  y\,\bigvee_{R\in\mn{sig}(\varphi_2)} ((Rxy\vee Ryx) \rightarrow
  \bigvee_{\lambda\in T}\lambda(x,y)) \big)
\end{align*}
is satisfiable. Note that each of the mentioned sentences is
formulated in GF$^2$ and at most single exponential in size (in the
size of $\vp_1$ and $\vp_2$), thus satisfiability can be decided in
2\ExpTime.

To build the automaton $\Amc_2$ from
Lemma~\ref{lem:bisi-automata}, set
$\Amc_2=(Q_2,\Theta,q_0,\delta_2,\Omega_2)$ where
$Q_2$ is 
$$
\begin{array}{l}
\{q_0,q_\bot\}\cup \mn{TP}_0 
\cup
\{t,t^?,t_\uparrow,t_\downarrow,t_\&,t^i,t^i_\uparrow,t^i_\downarrow\mid
t \in \mn{TP}_1, \; i \in \{ 0,1 \}\}
\, \cup \\[1mm]
\{ \lambda,\lambda_\uparrow,\lambda^i,\lambda^i_\uparrow \mid 
\lambda \in \mn{TP}_2,\; i \in \{ 0,1 \} \},
\end{array}
$$
$\Omega_2$ assigns two to all states except for those of the form
$t^?$, to which it assigns one.  

The automaton begins by choosing the $0$-type $s$ realized in the
forest model \Bmf of $\vp_2$ whose existence it aims to verify.  For
every $\exists x\varphi(x)\in s$, it then chooses a 1-type $t$ in
which $\varphi(x)$ is realized in \Bmf and sends off a copy of itself
to find a node where $t$ is realized. To satisfy Condition~1 of
Theorem~\ref{thm:charsimp}, at each node it further chooses a 1-type
that is compatible with $s$, to be realized at that node. This is
implemented by the following transitions:
$$ \begin{array}{rcll}
  \delta_2(q_0,\sigma) & = & \displaystyle\bigvee_{s\in \mn{TP}_0}
  \big(s\wedge \bigwedge_{\exists x\,\varphi(x)\in
  s}\bigvee_{\substack{t\in
    \mn{TP}_1 \mid \\ s\cup\{\varphi(x)\}\subseteq t}} t^?\big)\\[5mm]
  \delta_2(s,\sigma) & = & \Box s\wedge
  \displaystyle\bigvee_{\substack{t\in \mn{TP}_1,s\subseteq
    t}} t& 
\\[5mm]
  \delta_2(t^?,\sigma) & = & \langle-1\rangle t^? \vee
  \Diamond t^? \vee t^0
\end{array}
$$
where $s$ ranges over $\mn{TP}_0$.
When a state of the form $t$ is assigned to a node $w$, this is an
obligation to prove that there is a GF$^2(\Sigma)$-bisimulation
between the element $w$ in $\Amf_{(T,L)}$ and an element $b$ of type
$t$ in~\Bmf. A state of the form $t^0$ represents the obligation to
verify that there is an $A_\bot$-delimited GF$^2(\Sigma)$-bisimulation
between $w$ and an element of type $t$ in \Bmf. We first verify that
the former obligations are satisfied. This requires to follow all
successors of $w$ and to guess types of successors of $b$ to be mapped
there, satisfying the back condition of bisimulations. We also
need to guess successors of $b$ in $\Bmf$ (represented as a
neighborhood for $t$) to satisfy the existential demands of $t$ and
then select successors of $a$ to which they are mapped, satisfying the
``back'' condition of bisimulations.  Whenever we decide to realize a
1-type $t$ in \Bmf that does not participate in the bisimulation
currently being verified, we also send another copy of the automaton
in state $t^?$ to guess an $a \in A_{(T,L)}$ that we can use to
satisfy Condition~2 from Theorem~\ref{thm:charsimp}:
\pagebreak
$$
\begin{array}{rcll}
  \delta_2(t,(\tau,M)) & = & t_\uparrow \wedge \Box t_\downarrow\wedge \displaystyle
  \bigvee_{T \mid t\approx
  T} \bigwedge_{\lambda\in
  T}( \Diamond \lambda \vee \lambda_\uparrow )& \text{if
  }\tau_y=_\Sigma t 
\end{array}
$$
$$
\begin{array}{rcll}
  \delta_2(t,(\tau,M)) & = & \mn{false} & \text{if }\tau_y\not=_\Sigma 
t 
\\[2mm]
  \delta_2(t_\downarrow,(\tau,M)) & = & \mn{true} & \text{if $\tau$ is
  not $\Sigma$-guarded}\\
  \delta_2(t_\downarrow,(\tau,M)) & = &
  \displaystyle\bigvee_{\lambda  \mid t\approx \lambda \wedge
  \tau=_\Sigma \lambda}\lambda_y & \text{if
  $\tau$ is $\Sigma$-guarded}\\[2mm]
  \delta_2(t_\uparrow,(\tau,M)) & = & \mn{true} & \text{if $\tau$ is 
  not $\Sigma$-guarded}\\
  \delta_2(t_\uparrow,(\tau,M)) & = &
  \displaystyle\bigvee_{\lambda  \mid t\approx \lambda \wedge
  \tau=_\Sigma \lambda^-} [-1] \lambda_y & \text{if
  $\tau$ is $\Sigma$-guarded}\\[2mm]
  \delta_2( \lambda,(\tau,M)) & = & \lambda_y & \text{if }\lambda\text{
  is $\Sigma$-guarded and } \tau=_\Sigma\lambda  \\
  \delta_2( \lambda,(\tau,M)) & = & \mn{false} & \text{if }\lambda\text{
  is $\Sigma$-guarded and } \tau\not=_\Sigma\lambda  \\
  \delta_2( \lambda,(\tau,M)) & = & \lambda_y^? & \text{if $\lambda$ is
  not $\Sigma$-guarded} \\[2mm]
  \delta_2( \lambda_\uparrow,(\tau,M)) & = & \langle -1 \rangle \lambda_y  & \text{if }\lambda\text{
  is $\Sigma$-guarded and } \tau=_\Sigma\lambda^-  \\
  \delta_2( \lambda_\uparrow,(\tau,M)) & = & \mn{false} & \text{if }\lambda\text{
  is $\Sigma$-guarded and } \tau\not=_\Sigma\lambda^-  \\
  \delta_2( \lambda_\uparrow,(\tau,M)) & = & \lambda_y^? & \text{if $\lambda$ is
  not $\Sigma$-guarded}
\end{array}
$$
where $\tau_y =_\Sigma t$ means that the atoms in $\tau$ that mention
only $y$ are identical to the $\Sigma$-relational atoms in $t$ (up to
renaming $x$ to $y$), $\tau=_\Sigma\lambda$ means that the restriction
of $\lambda$ to $\Sigma$-atoms is exactly $\tau$, and $\lambda^-$ is
obtained from $\lambda$ by swapping $x$ and $y$. We need further
transitions to satisfy the obligations represented by states of the
form $t^0$, which involves checking $A_\bot$-delimited bisimulations.
Details are given in the appendix where also the correctness of the 
construction is proved.
\begin{theorem}
\label{thm:upper}
In GF$^{\,2}$, $\Sigma$-entailment and conservative
extensions can be decided in time $2^{2^{p(|\vp_2| \cdot \log|\vp_1|)}}$,
for some polynomial $p$.
Moreover, $\Sigma$-inseparability is in \textnormal{2}\ExpTime.
\end{theorem}
Note that the time bound for conservative extensions given in
Theorem~\ref{thm:upper} is double exponential only in the size of
$\vp_2$ (that is, the extension). In ontology engineering
applications, $\vp_2$ will often be small compared with $\vp_1$.

\smallskip

A matching lower bound is proved using a reduction of the word problem
of exponentially space-bounded alternating Turing machines, see the
appendix for details. The construction is inspired
by the proof from \cite{GhilardiLutzWolter-KR06} that conservative
extensions in the description logic \ALC are 2\ExpTime-hard, but the
lower bound does not transfer directly since we are interested here in
witness sentences that are formulated in GF$^2$ rather than in \ALC.
\begin{theorem}
\label{thm:lower}
In any fragment of FO that contains GF$^{\,2}$, 
the problems 
 $\Sigma$-entailment, $\Sigma$-inseparability, 
conservative extensions, and strong $\Sigma$-entailment are \textnormal{2}\ExpTime-hard.
\end{theorem}

\section{Conclusion}

We have shown that conservative extensions are undecidable in
(extensions of) GF and FO$^2$, and that they are decidable and
2\ExpTime-complete in GF$^2$. It thus appears that decidability of
conservative extensions is linked even more closely to the tree model
property than decidability of the satisfiability problem: apart from
cycles of length at most two, GF$^2$ enjoys a `true' tree model
property while GF only enjoys a bounded treewidth model property and
FO$^2$ has a rather complex regular model property that is typically
not even made explicit. As future work, it would be interesting to
investigate whether conservative extensions remain decidable when
guarded counting quantifiers, transitive relations, equivalence
relations, or fixed points are added, see e.g.\ \cite{DBLP:journals/logcom/Pratt-Hartmann07,DBLP:journals/iandc/Kieronski06,GradelW99}. It
would also be interesting to investigate a finite model version of
conservative extensions.

\bibliographystyle{plainurl}
\bibliography{icalp17}

\clearpage

\appendix

\section{Proofs for Section~\ref{sect:undec}}
We split the proof of Lemma~\ref{lem:gfundec1plus} into two parts.
\begin{lemma}\label{lem:gfundec1}
If $M$ halts then $\vp_{1} \wedge \vp_{2}$ is not a GF$^2$-conservative extension of $\vp_{1}$.
\end{lemma}
\begin{proof}
Assume that $M$ halts. We define a witness $\psi$ for non-conservativity.
It says that every element participates
in a substructure that represents the computation of $M$ on input $(0,0)$, that is: if the computation is
$(q_0,n_0,m_0),\dots,(q_k,n_k,m_k)$, then there is an $N$-path of
length $k$ (but not longer) such that any object reachable in
$i \leq k$ steps from the beginning of the path is labeled with
$q_i$, has an outgoing $R_0$-path of length $n_i$ and no longer
outgoing $R_0$-path, and likewise for $R_1$ and $m_i$. In more detail,
consider the $\Sigma$-structure $\Amf$ with
$$
A = \{0,\ldots,k\} \cup \{ a_{j}^{i} \mid 0<i\leq k, 0<j<n_{i}\} \cup \{ b_{j}^{i} \mid 0<i\leq k, 0<j< m_{i}\}
$$
in which
$$
\begin{array}{rcl}
  N^{\Amf} &=& \{ (i,i+1) \mid i<k\} \\[1mm]
  R_{1}^{\Amf} &=& \bigcup_{i\leq k}\{ (i,a_{1}^{i}),
                   (a_{1}^{i},a_{2}^{i}),\ldots,(a_{n_{i}-2}^{i},a_{n_{i}-1}^{i})\} \\[1mm]
  R_{2}^{\Amf} &=& \bigcup_{i\leq k}\{ (i,b_{1}^{i}),
                   (b_{1}^{i},b_{2}^{i}),\ldots,(b_{m_{i}-2}^{i},b_{m_{i}-1}^{i})\} \\[1mm]
  S^{\Amf} &=& \{0\} \\[1mm]
  q^{\Amf} &=& \{ i \mid q_{i}=q\} \text{ for any } q\in Q.
\end{array}
$$
Then let $\psi$ be a $GF^{2}(\Sigma)$-sentence that describes $\Amf$
up to global GF$^{2}(\Sigma)$-bisimulations.
Clearly $\Amf$ satisfies $\varphi_1 \wedge \psi$.
It thus remains to show that $\vp_1 \wedge \vp_2 \wedge \psi$ is unsatisfiable.
But this is clear as there are no $N$-paths of length $>k$ in any model of $\psi$
and since there are no defects in register updates in any model of $\psi$.
\end{proof}
\begin{lemma}\label{lem:gfundec2}
If there exists a $\Sigma$-structure that satisfies $\vp_{1}$ and cannot be extended
to a model of $\vp_{2}$, then $M$ halts.
\end{lemma}
\begin{proof}
Let $\Amf$ be a $\Sigma$-structure satisfying $\vp_{1}$ that cannot be extended
to a model of $\vp_{2}$. Then $S^{\Amf}\not=\emptyset$ and there exists an $N$-path labeled
with states in $Q$ starting in $S$. Since $\Amf$ cannot be extended to a model of $\vp_{2}$
the computation starting from $S$ is finite. Moreover, one can readily prove by induction
that no register update defects occur since otherwise $\vp_{2}$ can be satisfied.
\end{proof}
We now prove Theorem~\ref{mainth0fo2} using a reduction of an undecidable tiling problem.
\begin{definition}
  A \emph{tiling system} $\mathfrak{D}=(\mathfrak{T},H,V,{\sf Right},{\sf Left},{\sf Top},{\sf Bottom})$
 consists of a finite set $\mathfrak{T}$ of
 \emph{tiles}, horizontal and vertical \emph{matching relations} $H,V
  \subseteq \Tmf \times \Tmf$, and subsets ${\sf Right}$, ${\sf Left}$, ${\sf Top}$, and ${\sf Bottom}$
   of $\Tmf$ containing the \emph{right} tiles, \emph{left} tiles, \emph{top} tiles, and
  \emph{bottom} tiles, respectively. A \emph{solution} to $\Dmf$ is a triple
  $(n,m,\tau)$ where $n,m \in \Nbbm $ and $\tau: \{0,\ldots,n-1\}
  \times \{0,\ldots,m-1\} \rightarrow \Tmf$ such that the following hold:
\begin{enumerate}

\item $(\tau(i,j),\tau(i+1,j)) \in H$, for all $i<n$ and $j \leq m$;

\item $(\tau(i,j),\tau(i,j+1)) \in V$, for all $i\leq n$ and $j<m$;

\item $\tau(0,j) \in {\sf Left}$ and $\tau(n,j) \in {\sf Right}$, for all $j \leq m$;

\item $\tau(i,0) \in {\sf Bottom}$ and $\tau(i,m) \in {\sf Top}$, for all $i \leq n$.

\end{enumerate}

\end{definition}

We show how to convert a tiling system $\Dmf$ into FO$^2$-sentences
$\vp_1$ and $\vp_2$ such that $\Dmf$ has a solution iff $\vp_1 \wedge
\vp_2$ is not a conservative extension of $\vp_1$.  In particular,
models of witness sentences will define solutions of $\Dmf$.

\medskip

Let $\Dmf=(\Tmf,H,V,{\sf Right},{\sf Left},{\sf Top},{\sf Bottom})$ be a tiling system. The formula $\vp_1$ uses
the following set $\Sigma$ of predicates:
\begin{itemize}
\item binary predicates $R_h$ and $R_v$ (representing a grid) and $S_h$ and $S_v$ (for technical reasons),
\item unary predicates $T$ for every $T\in \mathfrak{T}$, $G$ (for the
  domain of the grid), $O$ (for the lower left corner of the grid),
  $B_\rightarrow$,  $B_\leftarrow$,  $B_\uparrow$, and  $B_\downarrow$
(for the borders of the grid).
\end{itemize}
Then $\vp_1$ is the conjunction of the following sentences:
\begin{enumerate}
\item Every position in the $n \times m$ grid is labeled with exactly
one tile and the matching conditions are satisfied:
$$
\begin{array}{l}
\displaystyle  
\forall x \big(Gx \rightarrow  \bigvee_{T\in \Tmf}(Tx \wedge
  \bigwedge_{T' \in \Tmf,\; T'\not=T} \neg T'x)\big) \\[5mm]
\displaystyle  
\forall x \big(Gx \rightarrow  \bigwedge_{T\in \Tmf}\big(Tx \rightarrow
(\bigvee_{(T,T') \in H} \forall y (R_hxy \rightarrow T'y) \wedge \bigvee_{(T,T') \in V}
\forall y (R_vxy \rightarrow T'y))\big)\big).
\end{array}
$$
%
\item The predicates $B_\rightarrow$, $B_\leftarrow$, $B_\uparrow$,
  and $B_\downarrow$ mark the borders of the grid:
$$
\begin{array}{l}
\forall x\big(Gx \wedge B_\rightarrow x \rightarrow \big(\neg\exists y R_hxy \wedge \forall y (R_vxy \rightarrow B_\rightarrow y)
                                                                                   \wedge
  \forall y (R_vyx \rightarrow B_\rightarrow y)\big)\big) \\[2mm]
\forall x\big(Gx \wedge \neg B_\rightarrow x \rightarrow \exists y R_hxy\big)
\end{array}
$$
and similarly for $B_\leftarrow$, $B_\uparrow$, and $B_\downarrow$.

\item There is a grid origin:
  $$
  \exists x \, (Ox \wedge B_\leftarrow x \wedge B_\downarrow x).
  $$

\item The grid elements are marked by $G$:
$$
\forall x \big(Ox \rightarrow Gx), \quad  \forall x (Gx \rightarrow \forall y (R_hxy \rightarrow Gy)),\quad
\forall x (Gx \rightarrow \forall y (R_vxy \rightarrow Gy)).
$$
\item The tiles on border positions are labeled with appropriate tiles:
$$
\forall x (B_\rightarrow x \rightarrow \bigvee_{T\in {\sf Right}}T(x)).
$$
and similarly for $B_\leftarrow$, $B_\uparrow$,  and $B_\downarrow$.
\item The predicates $S_h$ and $S_v$ occur in $\vp_{1}$: any FO$^2$-tautology using them.
\end{enumerate}
This finishes the definition of $\vp_{1}$.
The sentence $\vp_2$ introduces two new unary predicates $Q$ and~$P$ and is the conjuntion of $\exists x(Ox \wedge Qx)$
and
%
$$
\forall x \big(Qx \rightarrow \big(\exists y (R_hxy \wedge Qy) \vee \exists y(R_vxy \wedge Qy)
\vee \vp_Dx\big)\big)
$$
where
$$
\vp_Dx= \exists y \big( R_hxy \wedge \forall x (R_vyx \rightarrow Px)\big)
\wedge \exists y \big( R_vxy \wedge \forall x (R_hyx \rightarrow \neg Px)\big)
$$
Thus, $\vp_D$ describes a defect in the grid: there exist an $R_h$-successor $y_{1}$ and an $R_v$-successor $y_{2}$ of $x$ such that
every $R_v$-successor of $y_{1}$ is labeled with $P$ and every $R_h$-successor of $y_{2}$ is labeled with $\neg P$.
Informally, we can satisfy $\varphi_2$ only if from some element of $O$, there
is an infinite $R_h/R_v$-path or a non-closing grid cell can be reached by a finite such path.
To make this precise, we introduce some notation. Let $\Sigma'\supseteq \Sigma$ and let $\Bmf$ be a $\Sigma'$-structure.
Then the $\Sigma$-structure $\Amf$ obtained from $\Bmf$ by removing the interpretation of predicates in $\Sigma'\setminus\Sigma$
is called the \emph{$\Sigma$-reduct of $\Bmf$} and $\Bmf$ is called a \emph{$\Sigma'\setminus\Sigma$-extension of $\Amf$}.
For a $\Sigma$-structure $\Amf$, 
we say
that \emph{$a\in A$ is the root of a non-closing grid cell} if there are $(a,b_{1})\in R_h^{\Amf}$ and
$(a,b_{2})\in R_v^{\Amf}$ such that there does not exist a $c\in A$ with $(b_{1},c)\in R_v^{\Amf}$
and $(b_{2},c)\in R_h^{\Amf}$. Now, it is straightforward to show the following characterization of $\vp_{2}$.
\begin{lemma}\label{lem:1fo2}
  $\vp_{2}$ can be satisfied in a $\{Q,P\}$-extension of a
  $\Sigma$-structure $\Amf$ iff there exists an element of $O^{\Amf}$
  that starts an infinite $R_h/R_v$-path or a finite $R_h/R_v$-path
  to a root of a non-closing grid cell.
\end{lemma}

We now argue that $\Dmf$ has a solution iff $\vp_{1} \wedge \vp_{2}$ is not a conservative extension of $\vp_{1}$.
\begin{lemma}\label{lem:2fo2}
If $\Dmf$ has a solution then $\vp_{1}\wedge\vp_{2}$ is not a FO$^2$-conservative extension of $\vp_{1}$.
\end{lemma}
\begin{proof}
Assume that $\Dmf$ has a solution $(n,m,\tau)$. We
define a witness $\psi$, first using additional
fresh unary predicates and then argueing that these can be
removed. Thus introduce fresh unary predicates $P_{i,j}$ for all $i<n$
and $j<m$. Intuitively, $P_{i,j}$ identifies grid position
$(i,j)$. Set
$$
\begin{array}{rcl}
  \psi &=& \forall x \, (Gx \rightarrow \bigvee_{i,j} P_{i,j}x) \\[1mm]
  && \displaystyle \wedge\; \bigwedge_{(i,j) \neq (i',j')}\forall x \, \neg (P_{i,j}x \wedge P_{i',j'}x) \\[1mm]
  && \wedge\; 
\forall x \forall y \, (R_hxy \leftrightarrow \bigvee_{i,j}
  P_{i,j}x \wedge P_{i+1,j}y) \\[1mm]
  && \wedge\; \forall x \forall y \, (R_vx,y \leftrightarrow \bigvee_{i,j}
  P_{i,j}x \wedge P_{i,j+1}y) \\[1mm]
  && \wedge\; \forall x \, (Ox \rightarrow P_{0,0}x).
\end{array}
$$
We first show that $\varphi_1 \wedge \psi$ is satisfiable. As the
model, simply take the standard $n \times m$-grid in which all
positions are labeled with $P_{i,j}$, $G$, $O$, $B_{\rightarrow}$ etc in the
expected way, and that is tiled according to $\tau$. It is easily
verified that this structure satisfies both $\varphi_1$ and $\psi$.
It thus remains to show that $\vp_1 \wedge \vp_2 \wedge \psi$ is
unsatisfiable. By Lemma~\ref{lem:1fo2} it suffices to show that there is no model $\Amf$ of $\vp_{1}\wedge \psi$ in
which there exists an element of $O^{\Amf}$ starting an infinite $R_h/R_v$-path or a finite $R_h/R_v$-path leading to a root of a
non-closing grid cell. Assume for a proof by contradiction that there exists
such a model $\Amf$ and $a\in O^{\Amf}$. Then we find a sequence $a_{0}R^{\Amf}_{z_{0}}a_{1}R^{\Amf}_{z_{1}}a_{2}\cdots$ with
$a_{0}=a$ and $z_{i}\in \{x,y\}$ such that either some $a_{h}$ is the root of a non-closing grid cell or the sequence is infinite.
By $\vp_{1}$ and the first conjunct of $\psi$ for each $a_{k}$ there exists $P_{i,j}$ with $a_{k}\in P_{i,j}^{\Amf}$.
By the last conjunct of $\psi$, $a_{0}\in P_{0,0}^{\Amf}$. By the remaining conjuncts of $\psi$ we have $k\geq i+j$ if $a_{k}\in P_{i,j}^{\Amf}$
and it follows that there is no $a_{k}$ with $k>n+m$. Thus, assume some $a_{k}$ is the root of a non-closing grid.
Then we have $(a_{k},b_{1})\in R_h^{\Amf}$ and $(a_{k},b_{2})\in R_v^{\Amf}$ such that there is no $c\in A$ with $(b_{1},c)\in R_v^{\Amf}$
and $(b_{2},c)\in R_h^{\Amf}$. By $\psi$, there are $i,j$ with $b_{1}\in P_{i+1,j}^{\Amf}$ and $b_{2}\in P_{i,j+1}^{\Amf}$.
By the second set of conjuncts of $\vp_{1}$ there exists $(b_{1},c)\in R_v^{\Amf}$. By $\psi$, $c\in P_{i+1,j+1}^{\Amf}$.
But then again using $\psi$ we obtain that $(b_{2},c)\in R_h^{\Amf}$ and we have derived a contradiction.

\smallskip As announced, we now show how to remove the additional
predicates $P_{i,j}$. To this end, we use the binary predicates $S_h,S_v$.
In the sentence $\psi$, we replace every
occurrence of $P_{i,j}(x)$ with a formula saying: there is an outgoing
$S_h$-path of length $i$, but not of length $i+1$ and an outgoing
$S_v$-path of length $j$, but not of length $j+1$. When we build a
model of $\varphi_1 \wedge \psi$, we now need to introduce
additional elements for the ``counting paths''. We make the predicate
$G$ false on all those elements and true everywhere on the grid.
\end{proof}
\begin{lemma}\label{lem:3fo2}
If there exists a $\Sigma$-structure that satisfies $\vp_{1}$ and cannot be extended to a model of $\vp_{2}$,
then $\Dmf$ has solution.
\end{lemma}
\begin{proof}
Take a $\Sigma$-structure $\Amf$ satisfying $\vp_{1}$ that cannot be extended to a model of $\vp_{2}$.
By the conjunct of $\vp_{1}$ given in Item~3, $O^{\Amf}\cap B_{\leftarrow}^{\Amf}\cap B_{\downarrow}^{\Amf}\not=\emptyset$.
Take $a\in O^{\Amf}\cap B_{\leftarrow}^{\Amf}\cap B_{\downarrow}^{\Amf}$. By Lemma~\ref{lem:1fo2}, $a$ does not start an infinite $R_h/R_v$-path or a
finite $R_h/R_v$-path leading to the root of a non-closing grid cell. Using the conjuncts of $\vp_{1}$ it is now
straightforward to read off a tiling from $\Amf$.
\end{proof}
Theorem~\ref{mainth0fo2} is an immediate consequence of Lemmas~\ref{lem:2fo2} and~\ref{lem:3fo2}.
%

Note that the sentences $\vp_{1}$ and
$\vp_{2}$ can be replaced by equivalent $\mathcal{ALC}$-TBoxes: in
$\vp_{2}$, we can replace the conjunct $\exists x (Ox \wedge Qx)$
which cannot be expressed in $\mathcal{ALC}$ by the concept inclusion
$\top \sqsubseteq \exists S.(O \sqcap Q)$ with $S$ a fresh binary
predicate. Consequently, FO$^{2}(\Sigma)$-inseparability of
$\mathcal{ALC}$-TBoxes is undecidable.
%
%
%






\section{Preliminaries on Bisimulations}
We first show that GF$^{2}(\Sigma)$-bisimulations characterize
the expressive power of GF$^{2}(\Sigma)$. The proofs are standard \cite{DBLP:books/daglib/p/Gradel014,GorankoOtto,ANvB98}.
We start by giving a precise definition of
$k$-GF$^{2}(\Sigma)$ bisimilarity. Let $\Amf$ and $\Bmf$ be structures, $a\in A$,
and $b\in B$. The definition is by induction on $k\geq 0$.
Then $(\Amf,a) \sim^{0}_\Sigma (\Bmf,b)$ iff $\text{at}_{\Amf}^{\Sigma}(a)= \text{at}_{\Bmf}^{\Sigma}(b)$
and $(\Amf,a) \sim^{k+1}_\Sigma (\Bmf,b)$ iff $\text{at}_{\Amf}^{\Sigma}(a)= \text{at}_{\Bmf}^{\Sigma}(b)$
and
\begin{enumerate}
\item for every $a'\not=a$ such that $\text{at}_{\Amf}^{\Sigma}(a,a')$ is guarded there exists
$b'\not=b$ such that $\text{at}_{\Amf}^{\Sigma}(a,a')= \text{at}_{\Bmf}^{\Sigma}(b,b')$ and
$(\Amf,a') \sim^{k}_\Sigma (\Bmf,b')$
\item for every $b'\not=b$ such that $\text{at}_{\Bmf}^{\Sigma}(b,b')$ is guarded there exists
$a'\not=a$ such that $\text{at}_{\Bmf}^{\Sigma}(b,b')= \text{at}_{\Amf}^{\Sigma}(a,a')$ and
$(\Amf,a') \sim^{k}_\Sigma (\Bmf,b')$.
\end{enumerate}
Denote by openGF$^2$ the fragment of GF$^2$ consisting of all GF$^2$ formulas with one free variable
in which equality is not used as a guard and which do not contain a subformula that is a sentence.
It is not difficult to prove the following result.
\begin{lemma}\label{lem:boolean}
  Every GF$^{\,2}$ sentence is equivalent to a Boolean combination
of sentences of the form $\forall x \vp(x)$, where $\vp(x)$ is a
openGF$^{\,2}$ formula.
\end{lemma}
A structure $\Amf$ is \emph{$\omega$-saturated} if for every finite set $\{a_{1},\ldots,a_{n}\}\subseteq A$
and every set $\Gamma(x)$ of FO formulas using elements of $\{a_{1},\ldots,a_{n}\}$ as constants the following holds:
if every finite subset of $\Gamma(x)$ is satisfiable in the structure $(\Amf,a_{1},\ldots,a_{n})$,
then $\Gamma(x)$ is satisfiable in $(\Amf,a_{1},\ldots,a_{n})$. For every structure $\Amf$ there exists an
elementary extension $\Amf'$ of $\Amf$ that is $\omega$-saturated~\cite{CK90}.
Mostly we only require a weaker form of saturation.
A structure $\Amf$ is \emph{successor-saturated} if for any $a\in A$ and set $\Gamma(x)$ of openGF$^2$
formulas the following holds for any atomic guarded binary type $\tau$: if for any
finite subset $\Gamma'$ of $\Gamma$ there exists $a'\not=a$ with $\text{at}_{\Amf}(a,a')=\tau$ and $\Amf\models \psi(a')$ for all 
$\psi\in \Gamma'$, then there exists $b'\not=a$ with $\text{at}_{\Amf}(a,b')=\tau$ and $\Amf\models \psi(b')$ for all $\psi\in \Gamma$.
Observe that structures of finite outdegree and $\omega$-saturated structures are successor-saturated.

The \emph{depth} of a GF$^2$
formula $\vp$ is the number of nestings of guarded quantifications in $\vp$. We first characterize
openGF$^2$. The proof is standard and omitted.
\begin{lemma}\label{localbisim}
Let $\Amf$ and $\Bmf$ be structures, $\Sigma$ a signature, and $a\in A$, $b\in B$.
\begin{enumerate}
\item The following conditions are equivalent for all $k\geq 0$:
\begin{itemize}
\item $\Amf\models \varphi(a)$ iff $\Bmf\models \varphi(b)$ holds for
  all openGF$^{2}(\Sigma)$
formulas $\varphi(x)$ of depth~$k$;
\item $(\Amf,a) \sim^{k}_{\Sigma} (\Bmf,b)$.
\end{itemize}
\item If $(\Amf,a) \sim_{\Sigma} (\Bmf,b)$, then $\Amf\models \varphi(a)$ iff $\Bmf\models \varphi(b)$ holds
for all openGF$^{2}(\Sigma)$ formulas $\varphi(x)$. The converse
direction holds if $\Amf$ and $\Bmf$ are successor-saturated.
\end{enumerate}
\end{lemma}
We also require the following link between bounded bisimulations and unbounded bisimulations
which follows from Lemma~\ref{localbisim}.
%
\begin{lemma}\label{bisim}
Let $\Amf$ and $\Bmf$ be successor-saturated structures, $a\in A$, and $b\in B$.
If $(\Amf,a)\sim^k_\Sigma(\Bmf,b)$ for all $k\geq 0$, then $(\Amf,a) \sim_\Sigma(\Bmf,b)$.
\end{lemma}
We now consider `global' versions of the bounded bisimulations introduced above to characterize
GF$^2$. Call structures $\Amf$ and $\Bmf$ \emph{globally $k$-GF$^{2}(\Sigma)$-bisimilar} if for all
$a\in A$ there exists $b\in B$ such that $(\Amf,a) \sim^{k}_{\Sigma} (\Bmf,b)$
and, conversely, for every $b\in B$ there exists $a\in A$ with $(\Amf,a) \sim^{k}_{\Sigma} (\Bmf,b)$.
$\Amf$ and $\Bmf$ are \emph{globally finitely GF$^{2}(\Sigma)$-bisimilar} iff they are globally
$k$-GF$^{2}(\Sigma)$-bisimilar for all $k\geq 0$. 
The following characterization result now follows from Lemma~\ref{lem:boolean} and Lemma~\ref{localbisim}.
\begin{lemma}\label{lem:global}
Let $\Amf$ and $\Bmf$ be structures and $\Sigma$ a signature.
\begin{enumerate}
\item The following conditions are equivalent:
\begin{itemize}
\item $\Amf\models \varphi$ iff $\Bmf\models \varphi$ holds for all
  GF$^{2}(\Sigma)$ sentences $\vp$;
\item $\Amf$ and $\Bmf$ are globally finitely GF$^{2}(\Sigma)$-bisimilar.
\end{itemize}
\item If $\Amf$ and $\Bmf$ are globally GF$^{2}(\Sigma)$-bisimilar, then $\Amf\models \varphi$ iff
$\Bmf\models \varphi$ holds for all GF$^{2}(\Sigma)$ sentences $\varphi$. The converse direction holds
if $\Amf$ and $\Bmf$ are $\omega$-saturated.
\end{enumerate}
\end{lemma}
Observe that in Lemma~\ref{lem:global} we cannot replace $\omega$-saturation by successor-saturation
or finite outdegree.

\section{Proofs for Section~\ref{sect:characterization}}

Based on the results presented in the previous section we prove the following characterization of $\Sigma$-entailment in FO$^2$.

\medskip
\noindent
{\bf Theorem~\ref{thm:char1}}
\emph{Let $\varphi_1,\varphi_2$ be GF$^2$-sentences and $\Sigma$ a
signature. Then $\vp_1 \models_\Sigma \vp_2$ iff for every model \Amf of $\vp_1$
of finite outdegree, there is a model \Bmf of $\vp_2$ such that
  \begin{enumerate}
  \item for every $a\in A$ there is a $b
    \in B$ such that $(\Amf,a) \sim_\Sigma (\Bmf,b)$
  \item for every $b\in B$ and every $k
    \geq 0$, there is an $a \in A$ such that $(\Amf,a) \sim^k_\Sigma (\Bmf,b)$.
  \end{enumerate}
}

\medskip
\noindent
\begin{proof}
  ``if''. Assume that for every model \Amf of $\vp_1$ of finite outdegree, there
  is a model \Bmf of $\vp_2$ as described in
  Theorem~\ref{thm:char1}. Take a $\Sigma$-sentence $\psi$ such that
  $\varphi_1 \wedge \psi$ is satisfiable. We have to show that
  $\varphi_2 \wedge \psi$ is satisfiable.
  We find a model \Amf of $\varphi_1 \wedge \psi$ that has finite outdegree. By
  assumption, there is a model \Bmf of $\varphi_2$ that satisfies
  Conditions~1 and~2 of Theorem~\ref{thm:char1}. It suffices to show
  that \Bmf satisfies $\psi$. But this follows from Lemma~\ref{lem:global}.

  \medskip

  ``only if''. Assume that $\varphi_1 \models_\Sigma \varphi_2$. Let
  \Amf be a model of $\varphi_1$ of finite outdegree. Let
  $\Gamma$ denote the set of all GF$^{2}(\Sigma)$ sentences $\psi$ with $\Amf\models \psi$.
  Then $\vp_{1} \wedge \bigwedge \Gamma'$ is satisfiable for every finite subset $\Gamma'$ of $\Gamma$.
  As $\varphi_1 \models_\Sigma \varphi_2$, $\vp_{2} \wedge \bigwedge \Gamma'$ is satisfiable for every
  finite subset $\Gamma'$ of $\Gamma$. By compactness $\{\vp_{2}\} \cup \Gamma$ is satisfiable. Then there
  exists an $\omega$-saturated model $\Bmf$ of $\{\vp_{2}\} \cup \Gamma$. By $\omega$-saturatedness,
  for every $a\in A$ there exists $b\in B$ such that $\Amf\models \vp(a)$ iff $\Bmf\models \vp(b)$ holds for all
  formulas $\vp(x)$ in openGF$^{2}(\Sigma)$. By Lemma~\ref{localbisim}, we have
  $(\Amf,a) \sim_{\Sigma} (\Bmf,b)$, as required for Condition~1. Condition~2 follows from
  Lemma~\ref{localbisim}.
%
\end{proof}
Before we come to the proof of Theorem~\ref{thm:charsimp} we prove another
characterization of $\Sigma$-entailment in GF$^2$.
If \Amf is a forest structure with $a,a' \in A$, then we write $a
\prec a'$ iff $a$ and $a'$ are part of the same $\Sigma$-tree in \Amf and $a$
is a ancestor of~$a'$ (recall that a $\Sigma$-tree in a forest structure $\Amf$ is a maximal
$\Sigma$-connected substructure of $\Amf$ and that we always assume a fixed root in trees within forest structures).
For \Amf and \Bmf structures and $a_\bot \in A$, an
\emph{$a_\bot$-delimited GF$^{2}(\Sigma)$-bisimulation
between \Amf and \Bmf} is defined like a GF$^{2}(\Sigma)$-bisimulation
except that Conditions~2 and~3 are not required to hold when
$a=a_\bot$. We indicate the existence of an $a_\bot$-delimited
bisimulation by writing $(\Amf,a) \sim^{a_\bot}_\Sigma (\Bmf,b)$.
This requires $a_\bot \preceq a$. 
We now give a characterization of $\Sigma$-entailment using
forest models in which we replace the bounded backward condition by an
unbounded condition.
%

\begin{theorem}\label{thm:character}
Let $\varphi_1,\varphi_2$ be GF$^2$-sentences and $\Sigma$ a
signature. Then $\vp_1 \models_\Sigma \vp_2$ iff for every regular
forest model \Amf of $\vp_1$ that has finite outdegree
there is a model \Bmf of $\vp_2$ such that
  \begin{enumerate}
  \item for every $a\in A$ there is a $b
    \in B$ such that $(\Amf,a) \sim_\Sigma (\Bmf,b)$
  \item for every $b\in B$, one of the following holds:
    \begin{enumerate}
    \item there is an $a \in A$ such that
      $(\Amf,a) \sim_\Sigma (\Bmf,b)$;
    \item there are $a_\bot,a_0,a_1,\dots,a'_0,a'_1,\dots\in A$ such
      that $a_\bot \prec a_0 \prec a_1 \prec \cdots$ and, for all $i
      \geq 0$, $a_i \prec a'_i$ and $(\Amf,a_i') \sim^{a_\bot}_\Sigma
      (\Bmf,b)$.
   \end{enumerate}
  \end{enumerate}
\end{theorem}
\begin{proof}
Using the proof of Theorem~\ref{thm:char1}, the fact that every (successor-saturated/finite outdegree)
structure $\mathfrak{A}$ can be unfolded into a globally GF$^{2}(\Sigma)$-bisimilar
(successor saturated/finite outdegree) forest model $\Bmf$, and the fact that,
consequently, every satisfiable GF$^2$ formula is satisfiable
in a regular forest model of finite outdegree one can easily
prove the following variant of Theorem~\ref{thm:char1} based on forest models:

\medskip
\noindent
{\bf Fact 1.} Let $\varphi_1,\varphi_2$ be GF$^2$-sentences and $\Sigma$ a
signature. Then $\vp_1 \models_\Sigma \vp_2$ iff for every regular
forest model \Amf of $\vp_1$ that has finite outdegree
there is a (successor saturated) forest model \Bmf of
$\vp_2$ such that
  \begin{enumerate}

  \item for every $a\in A$ there is a $b
    \in B$ such that $(\Amf,a) \sim_\Sigma (\Bmf,b)$

  \item for every $b\in B$ and every $k
    \geq 0$, there is an $a \in A$ such that $(\Amf,a)
    \sim^k_\Sigma (\Bmf,b)$.

  \end{enumerate}
To show Theorem~\ref{thm:character} it therefore suffices to show that for every regular
forest model \Amf of $\vp_1$ that has finite
outdegree and every successor-saturated forest model \Bmf of $\vp_2$,
Condition~2 in Fact~1 is equivalent to Condition~2 of Theorem~\ref{thm:character}.

  \smallskip

  Thus, let \Amf and \Bmf be as described.  The interesting direction is to prove
  that if Condition~2 in Fact~1 holds then Condition~2 of Theorem~\ref{thm:character} holds.
  Thus, assume that Condition~2 in Fact~1 holds. Take $b\in B$. We may assume it is a root $b$
  of a $\Sigma$-tree in \Bmf. Then there are $a_0,a_1,\dots \in A$
  such that for all $k$, $(\Amf,a_k) \sim^k_\Sigma (\Bmf,b)$. If
  infinitely many of the $a_i$ are identical, then there is an $a \in
  A$ such that $(\Amf,a) \sim^k_\Sigma(\Bmf,b)$ for all $k\geq 0$,
  thus $(\Amf,a) \sim_\Sigma(\Bmf,b)$ by Lemma~\ref{bisim} and we are
  done.  Therefore, assume that there are infinitely many distinct
  $a_i$. By `skipping' elements in the sequence $a_0,a_1,\dots$, we
  can then achieve that the $a_i$ are all distinct.

  \smallskip
  
 Two nodes $a,a' \in A$ are
\emph{downwards isomorphic}, written $a \sim_\downarrow a'$, if they
are the roots of isomorphic subtrees.
  For a forest structure \Amf, $a \in A$, and $i \geq 0,$ we denote by
  $\Amf|^{\uparrow i}_a$ the path structure obtained by restricting
  \Amf to those elements that can be reached from $a$ by traveling at
  most $i$ steps towards the root of the tree in \Amf that $a$ is part
  of (including $a$ itself). For $a,a' \in A$ and $i \geq 0$, we write
  $a \approx_i a'$ if there is an isomorphism $\iota$ from
  $\Amf|^{\uparrow i}_a$ to $\Amf|^{\uparrow i}_{a'}$ with
  $\iota(a)=a'$ such that $c \sim_\downarrow \iota(c)$ for all $c$. 
  Since \Amf is regular, \Amf contains
  only finitely many equivalence classes for each $\approx_i$.  By
  skipping $a_i$'s, we can thus achieve that
  \begin{itemize}

  \item[($*$)] $a_i \approx_k a_j$ for all $i,k,j$ with $k \leq i$ and $j
    >i$.

  \end{itemize}
  This also implies that each $a_i$ is at least $i$ steps away from
  the root of the tree in \Amf that it is in (since there are
  infinitely many~$a_i$, they must be unboundedly deep in their
  respective tree, and it remains to apply ($*$)). Let $c_i$ denote
  the element of $A$ reached from $a_i$ by traveling $i$ steps towards
  the root. Since \Amf is regular, there must be an infinite
  subsequence $a_{\ell_0},a_{\ell_1},\dots$ of $a_0,a_1,\dots$ such that
  $c_{\ell_i} \sim_\downarrow c_{\ell_j}$ for all $i,j$.

  \smallskip

  Choose some $a_\bot \in A$ with $a_\bot \sim_\downarrow c_{\ell_i}$
  for all $i$ (equivalently: for some $i$).  We can assume w.l.o.g.\
  that each $a_{\ell_i}$ is in the subtree rooted at $a_\bot$ and that
  when traveling ${\ell_i}$ steps from $a_{\ell_i}$ towards the root
  of the subtree that $a_\bot$ is in, then we reach exactly $a_\bot$.

  \smallskip

  Let $\Amf^*$ be the structure obtained in the limit of the
  neighborhoods
  $\Amf|^0_{a_{\ell_0}},\Amc|^1_{a_{\ell_1}},\dots$. That is, we start
  with the subtree of \Amf rooted at $a_{\ell_0}$, renaming
  $a_{\ell_0}$ to $a^*$, and then proceed as follows: after the $i$-th
  step, the constructed structure is isomorphic to the subtree of \Amf
  rooted at $a_\bot$ via an isomorphism that maps $a^*$ to
  $a_{\ell_i}$ and the root to $a_\bot$; by ($*$), we can thus add a
  path of predecessor to the root of the structure constructed so far,
  and then add additional subtrees to the nodes on the path as
  additional successors, making sure that the obtained structure is
  isomorphic to the subtree of \Amf rooted at $a_\bot$ via an
  isomorphism that maps $a^*$ to $a_{\ell_{i+1}}$ and the new root to
  $a_\bot$. By construction, $(\Amf^*,a^*) \sim_\Sigma^k (\Bmf,b)$
  for all $k \geq 0$ and thus Lemma~\ref{bisim} yields $(\Amf^*,a^*)
  \sim_\Sigma (\Bmf,b)$. 



%

  \smallskip

  Take some $a_{\ell_i}$. We aim to show that $(\Amf,a_{\ell_i})
  \sim^{a_\bot}_\Sigma
  (\Bmf,b)$. 
  Let $c$ be the element reached from $a^*$ in $\Amf^*$ by traveling
  $\ell_i$ steps upwards and recall that $a_\bot$ is the element reached
  from $a_{\ell_i}$ in $\Amf$ by traveling $\ell_i$ steps upwards. By
  construction of $\Amf^*$, we find an isomorphism from the subtree in
  $\Amf^*$ rooted at $c$ to the subtree in \Amf rooted at $a_\bot$
  that takes $c$ to $a_\bot$ and $a^*$ to $a_i$.  From $(\Amf^*,a^*)
  \sim_\Sigma (\Bmf,b)$, we thus obtain the desired $a_\bot$-delimited
  $\Sigma$-bisimulation that witnesses $(\Amf,a_i)
  \sim_\Sigma^{a_\bot} (\Bmf,b)$.

  \smallskip

  It remains to show the existence of the required elements
  $a'_0,a'_1,\dots$, that is, to show that there is a path through the
  subtree of $\Amf$ rooted at $a_\bot$ such that each $a_{\ell_i}$ is
  either on the path or can be reached by branching off at a different
  point of the path. This can be done in the following straightforward
  way. Starting at $a_\bot$, we define the path step by step. In every
  step, there must be at least one successor which is the root of a
  subtree that contains infinitely many $a_{\ell_i}$'s since \Amf has
  finite outdegree. We always proceed by choosing such a successor.
  This almost achieves the desired result, except that not all $a_{\ell_i}$
  are reachable from a \emph{distinct} node on the path by traveling
  downwards. However, there are infinitely many nodes on the path from
  which at least one $a_{\ell_i}$ can be reached by traveling downwards, so
  the problem can be cured by skipping $a_{\ell_i}$'s.
\end{proof}
We are now in a position to prove Theorem~\ref{thm:charsimp}. We require the following extended version
of $k$-GF$^2$-bisimilarity which respects the successor relation in forest structures.
Let $\Amf$ and $\Bmf$ be forest structures, $a\in A$,
and $b\in B$. The definition is by induction on $k\geq 0$.
Then $(\Amf,a) \sim^{0,\text{succ}}_\Sigma (\Bmf,b)$ iff $\text{at}_{\Amf}^{\Sigma}(a)= \text{at}_{\Bmf}^{\Sigma}(b)$
and $(\Amf,a) \sim^{k+1,\text{succ}}_\Sigma (\Bmf,b)$ iff $\text{at}_{\Amf}^{\Sigma}(a)= \text{at}_{\Bmf}^{\Sigma}(b)$
and
\begin{enumerate}
\item for every $a'\not=a$ such that $\text{at}_{\Amf}^{\Sigma}(a,a')$ is guarded there exists
$b'\not=b$ such that $\text{at}_{\Amf}^{\Sigma}(a,a')= \text{at}_{\Bmf}^{\Sigma}(b,b')$ and $b'$ is a successor of $b$ in $\Bmf$
iff $a'$ is a successor of $a$ in $\Amf$ and
$(\Amf,a') \sim^{k,\text{succ}}_\Sigma (\Bmf,b')$
\item for every $b'\not=b$ such that $\text{at}_{\Amf}^{\Sigma}(b,b')$ is guarded there exists
$a'\not=a$ such that $\text{at}_{\Amf}^{\Sigma}(b,b')= \text{at}_{\Bmf}^{\Sigma}(a,a')$ and $a'$ is a successor of $a$ in $\Amf$
iff $b'$ is a successor of $b$ in $\Bmf$ and 
$(\Amf,a') \sim^{k,\text{succ}}_\Sigma (\Bmf,b')$.
\end{enumerate}

\medskip
\noindent
{\bf Theorem~\ref{thm:charsimp}}
\emph{Let $\varphi_1,\varphi_2$ be GF$^2$-sentences and $\Sigma$ a
  signature. Then $\vp_1 \models_\Sigma \vp_2$ iff for every regular
  forest model \Amf of $\vp_1$ 
  that has finite outdegree and every set $A_{\bot}\subseteq A$
  with $A_{\bot}\cap \rho$ infinite for any infinite $\Sigma$-path $\rho$ in $\Amf$
  there is a model \Bmf of $\vp_2$ such that
  \begin{enumerate}
  \item for every $a \in A$, there is a $b \in B$ such that $(\Amf,a) \sim_\Sigma (\Bmf,b)$
  \item for every 1-type $t$ for $\vp_2$ that is realized in \Bmf,
    there are $a \in A$ and $b \in B$ such that $\mn{tp}_\Bmf(b)=t$
    and $(\Amf,a)\sim_{\Sigma}^{A_{\bot}}(\Bmf,b)$.
\end{enumerate} 
}

\medskip
\noindent
\begin{proof}
($\Leftarrow$) It suffices to show that for every $m>0$ and every regular forest model $\Amf$
 of $\vp_1$ that has finite outdegree there exists
 a model $\Bmf$ of $\vp_{2}$ such that $\Amf$ and $\Bmf$ are globally
 $m$-GF$^{2}(\Sigma)$-bisimilar.
 Assume $m>0$ and a regular forest model $\Amf$
 of $\vp_1$ that has finite outdegree is given. 
 Let $m'$ be the maximum of $m$ and the guarded quantifier depth of $\vp_{2}$. 
 Then $f(m,\vp_{2})$ denotes the maximal number of nodes in any $\Sigma\cup {\sf sig}(\vp_{2})$-forest model $\Cmf$
 which are pairwise $\sim^{m',\text{succ}}$-incomparable.
 Define $A_{\bot}\subseteq A$ on every $\Sigma$-tree with root $r$ in $\Amf$ in such a way that 
 $a\in A_{\bot}$ iff the distance between $r$ and $a$ is $kf(m,\vp_{2})$ for some $k\geq 0$.
 Let $\Bmf$ be a forest shaped model of $\vp_{2}$ satisfying the conditions of Theorem~\ref{thm:charsimp}.
 One can easily modify $\Bmf$ in such a way that in addition to the conditions given in the theorem 

\medskip
\noindent 
$(\ast)$ every 1-type $t$ for $\vp_{2}$ that is realized in $\Bmf$ is realized in the
root of a $\Sigma$-tree in $\Bmf$ and for every root $r$ of a $\Sigma$-tree in $\Bmf$
there exists $a \in A$ such that $(\Amf,a)\sim_{\Sigma}^{A_{\bot}}(\Bmf,r)$.

\medskip
To show $(\ast)$ first pick for every $a\in A$ a $b\in B$ with $(\Amf,a) \sim_\Sigma (\Bmf,b)$.
Let $S_{1}$ be the set of $b$'s just picked and let $\Bmf_{1}$ be the disjoint union of the structures 
induced in $\Bmf$ by the $\Sigma$-trees whose roots are in $S_{1}$. Next pick for every
1-type $t$ for $\vp_{2}$ that is realized in $\Bmf$ a $b\in B$ that realizes $t$. 
Let $S_{2}$ be the set of $b$'s just picked and let $\Bmf_{2}$ 
be the disjoint union of the structures 
induced in $\Bmf$ by the $\Sigma$-trees whose roots are in $S_{2}$. Finally, we add (recursively) witnesses for
guarded existential quantifiers not involving binary predicates from $\Sigma$ to the disjoint union
$\Bmf'$ of $\Bmf_{1}$ and $\Bmf_{2}$. In detail, take for any $b$ in $\Bmf'$ its copy $b'$ in $\Bmf$
and assume $c'$ in $\Bmf$ is such that $\{ R \mid (b',c')\in R^{\Bmf}$ or $(c',b')\in R^{\Bmf}\}$
is non-empty and contains no predicate in $\Sigma$. Then add to $\Bmf'$ a copy of the $\Sigma$-tree in $\Bmf'$
whose root $c$ realizes the same 1-type for $\vp_{2}$ as $c'$ and connect $c$ to $b$ by adding for all binary predicates $R$
the pair $(b,c)$ to the extension of $R$ if $(b',c')\in R^{\Bmf}$ and the pair $(c,b)$ to the extension of $R$ if
$(c',b')\in R^{\Bmf}$. We apply this procedure recursively to the new structure (in a fair way) and obtain the desired structure as the 
limit of the resulting sequence of structures.

We now modify $\Bmf$ in such a way that the resulting structure is still a model of $\vp_{2}$ but
in addition globally $m$-GF$^{2}(\Sigma)$-bisimilar to $\Amf$. Consider the structure $\Bmf_{r}$ induced by the $\Sigma$-tree 
with root $r$ in $\Bmf$.  
If there exists an $a\in A$ with $(\Amf,a) \sim_\Sigma (\Bmf,r)$ 
then we do not modify $\Bmf_{r}$ and set $\Bmf_{r}^{u}=\Bmf_{r}$. If no such $a$ exists, then we modify 
$\Bmf_{r}$ in such a way that every $b$ in the resulting $\Sigma$-tree
is $m$-GF$^{2}(\Sigma)$-bisimilar to some $a\in A$. 
Note that we only know that there exists $a\in A$ 
such that $(\Amf,a)\sim_{\Sigma}^{A_{\bot}}(\Bmf,r)$. By construction of $A_{\bot}$ this implies that $(\Amf,a)$ and
$(\Bmf,r)$ are $f(m,\vp_{2})$-GF$^{2}(\Sigma)$-bisimilar. Thus, it suffices to modify $\Bmf_{r}$ in such a way
that every node $b$ in the $\Sigma$-tree becomes $m$-GF$^{2}(\Sigma)$-bisimilar to some
$b'$ in the original $\Bmf_{r}$ with distance $\leq f(m,\vp_{2})-m$ from $r$.
To ensure that $\vp_{2}$ is still satisfied we make sure that the following stronger condition holds:
every node $b$ in the $\Sigma$-tree rooted at $r$ is $m'$-GF$^2$-bisimilar to some
$b'$ in the original $\Bmf_{r}$ with distance $\leq f(m,\vp_{2})-m'$ from $r$.
The construction is by a standard pumping argument. 
For $a,b\in B$ we say that \emph{$a$ blocks $b$} if $a\prec b$ and $(\Bmf,a)\sim^{m',\text{succ}}(\Bmf, b)$ and there is no $b' \prec b$ 
such that there is an $a'$ with $a'\prec b'$ and $(\Bmf,a')\sim^{m',\text{succ}}(\Bmf,b')$. The universe $B_{r}^{u}$ of $\Bmf_{r}^{u}$ is 
the set of words $a_{0}\cdots a_{n}$ with $a_{0},\ldots,a_{n}$ in $\Bmf_{r}^{u}$ and $a_{0}=r$ such that either $a_{i+1}$ is a 
successor of $a_{i}$ or there is a successor $b_{i+1}$ of $a_{i}$ such that $a_{i+1}$ blocks $b_{i+1}$.
Let $\text{tail}(a_{0}\ldots a_{n})=a_{n}$.
For every unary $R$ and $w\in A_{r}^{u}$ we set $w\in R^{\mathfrak{A}_{r}^{u}}$ if $\text{tail}(w)\in R^{\mathfrak{A}}$
and for every binary $R$ we set for $w\in A_{r}^{u}$: $(w,w)\in R^{\mathfrak{A}_{r}^{u}}$ if 
$(\text{tail}(w),\text{tail}(w))\in R^{\mathfrak{A}}$ and for $wb\in A_{r}^{u}$:
\begin{itemize}
\item $(w,wb)\in R^{\mathfrak{A}_{r}^{u}}$ if $(\text{tail}(w),b)\in R^{\mathfrak{A}}$ or there
is an $a$ such that $b$ blocks $a$ and $(\text{tail}(w),a)\in R^{\mathfrak{A}}$;
\item $(wb,w) \in R^{\mathfrak{A}_{a}^{u}}$ if $(b,\text{tail}(w))\in R^{\mathfrak{A}}$
or there is an $a$ such that $b$ blocks $a$ and $(a,\text{tail}(w))\in R^{\mathfrak{A}}$.
\end{itemize}
We now replace $\Bmf_{r}$ by $\Bmf_{r}^{u}$ in $\Bmf$. In more detail, take the disjoint union $\Bmf^{d}$ 
of all $\Bmf_{r}^{u}$, $r$ the root of a $\Sigma$-tree in $\Bmf$. Then add (recursively) witnesses for
guarded existential quantifiers not involving binary predicates from $\Sigma$ to $\Bmf^{d}$:
take for any $w$ in $\Bmf_{r}^{u}$ and any 1-type $t$ for $\vp_{2}$ that is realized in some node $c$ in
$\Bmf$ such that $\{ R \mid (\text{tail}(w),c)\in R^{\Bmf}$ or $(c,\text{tail}(w))\in R^{\Bmf}\}$
is non-empty and contains no predicate in $\Sigma$ the root $r'$ of a structure $\Bmf_{r'}^{u}$ such that $r'$ realizes
$t$ in $\Bmf_{r'}^{u}$. Then add to $\Bmf^{d}$ a new copy of $\Bmf_{r'}^{u}$ and connect $r'$ to $b$ 
by adding for any binary predicate $R$ the pair $(r,r')$ to $R^{\Bmf_{d}}$ if $(\text{tail}(w),c)\in R^{\Bmf}$ and 
the pair $(r',r)$ to $R^{\Bmf_{d}}$ if $(c,\text{tail}(w))\in R^{\Bmf}$.
We apply this procedure recursively to the new structure (in a fair way) and obtain the desired structure $\Bmf'$
as the limit of the resulting sequence of structures.
  
\medskip

($\Rightarrow$) Assume that $\vp_1 \models_\Sigma \vp_2$. Let $\Amf$ be a regular
  forest model \Amf of $\vp_1$ that has finite outdegree and let $A_{\bot}\subseteq A$ be such that 
 $A_{\bot}\cap \rho$ is infinite for any maximal infinite $\Sigma$-path $\rho$ in $\Amf$.
By Theorem~\ref{thm:character}, there is a model
\Bmf of $\vp_2$ such that
  \begin{enumerate}
  \item for every $a\in A$ there is a $b
    \in B$ such that $(\Amf,a) \sim_\Sigma (\Bmf,b)$
  \item for every $b\in B$, one of the following holds:
    \begin{enumerate}
    \item there is an $a \in A$ such that
      $(\Amf,a) \sim_\Sigma (\Bmf,b)$;
    \item there are $a_\bot,a_0,a_1,\dots,a'_0,a'_1,\dots\in A$ such
      that $a_\bot \prec a_0 \prec a_1 \prec \cdots$ and, for all $i
      \geq 0$, $a_i \prec a'_i$ and $(\Amf,a_i') \sim^{a_\bot}_\Sigma
      (\Bmf,b)$.
   \end{enumerate} 
\end{enumerate}
Let $t$ be a 1-type for $\vp_{2}$ realized by some $b \in B$. 
We have to find an $a\in A$ such that $(\Amf,a) \sim^{A_\bot}_\Sigma (\Bmf,b)$.
If there is an $a \in A$ such that
$(\Amf,a) \sim_\Sigma (\Bmf,b)$ then we are done as 
$(\Amf,a) \sim_{\Sigma}^{A_{\bot}} (\Bmf,b)$ follows.
Otherwise there are $a_\bot,a_0,a_1,\dots,a'_0,a'_1,\dots\in A$ such
      that $a_\bot \prec a_0 \prec a_1 \prec \cdots$ and, for all $i
      \geq 0$, $a_i \prec a'_i$ and $(\Amf,a_i') \sim^{a_\bot}_\Sigma
      (\Bmf,b)$. Then let $\rho$ be a $\Sigma$-path containing $a_\bot, a_0, a_{1},\ldots$.
$A_{\bot}\cap \rho$ is infinite and so we can choose an $a_{i}'$ such that there are at least two elements of
$A_{\bot}$ on the path from $a_{\bot}$ to $a_{i}'$. It follows from the definition of $\sim^{A_\bot}_\Sigma$
that $(\Amf,a_{i}') \sim^{A_\bot}_\Sigma (\Bmf,b)$, as required.
\end{proof}

\section{Proofs for Section~\ref{sect:upper}}

We construct the required 2ATAs.
\begin{lemma}\label{thm:gf2-automata}
  Let $\vp_1$ be a GF$^2$-sentence.
  There is a 2ATA $\Amc_1$ that accepts a $\Theta$-labeled tree
  $(T,L)$ iff $\Amf_{(T,L)}$ is a model of
  $\varphi_1$.
%
\end{lemma}
We assume that in all subformulas of $\vp_1$ of the form $\exists \ybf
(\alpha(\xbf,\ybf) \wedge \vp(\xbf,\ybf))$ and $\forall \ybf
(\alpha(\xbf,\ybf) \rightarrow \vp(\xbf,\ybf))$, \ybf consists of
exactly one variable and $\alpha(\xbf,\ybf)$ is a relational atom with
two variables or an equality
atom. 
This can be done w.l.o.g.\ because each sentence $\exists xy
\varphi(x,y)$ can be rewritten into $\exists x (x=x \wedge \exists y
\varphi(x,y))$, each sentence $\exists x (\alpha(x) \wedge
\varphi(x))$ with $\alpha$ a relational atom can be rewritten into
$\exists x (x=x \wedge \alpha(x) \wedge \varphi(x))$, and likewise for
universal quantifiers.
We further assume that $\varphi_1$
has no subformulas of the form $\exists x (x=y \wedge \varphi(x,y))$
with $x \neq y$; such formulas are equivalent to $\vp[y/x]$, that is,
the result of replacing in $\vp$ all occurrences of $x$ with $y$.
The result of these assumptions is that each
formula $\exists \ybf (\alpha(\xbf,\ybf) \wedge \vp(\xbf,\ybf))$ takes
the form $\exists x(x=x \wedge \psi(x))$ or $\exists x \psi(x,y)$, and likewise for
universally quantified formulas.
We
define $\Amc_1 = (Q_1,\Theta,q_{\vp_1},\delta_1,\Omega_1)$ where
\newcommand{\cvar}[1]{\underline{#1}}
$$
  Q_{1} = \{ q_{\varphi(x)} \mid \varphi(x) \in
  \mn{cl}(\vp_1) \}\cup
  \{ q_{\varphi(x,\cvar{y})}, q_{\varphi(\cvar{x},y)} \mid
  \vp(x,y) \in \mn{cl}(\vp_1) \}
$$
and $\Omega_1$ assigns two  to all
states except those of the form $q_{\exists x (x=x \wedge \psi(x))}$,
to which it assigns one. The underlining in states of the form
$q_{\varphi(x,\cvar{y})}$ and $q_{\varphi(\cvar{x},y)}$ serves as a
marking of the variable that is bound to the tree node to which the
state is assigned. We define the transition function $\delta_1$ as
follows, for each $\sigma=(\tau,M)$:
$$
\begin{array}{r@{\;}c@{\;}l}
  \delta_1(q_{Az},\sigma) &=& \left \{
    \begin{array}{ll}
\mn{true} & \text{if } Ay \in
  \tau \\ 
\mn{false} & \text{otherwise}
    \end{array}
\right. \\[3mm]
  \delta_1(q_{\neg Az},\sigma) &=& \left \{
    \begin{array}{ll}
\mn{true} & \text{if } Ay \notin
  \tau  \\ 
\mn{false} & \text{otherwise}
    \end{array}
\right. \\[3mm]
  \delta_1(q_{\varphi(z) \circ \psi(z)},\sigma) &=& q_{\varphi(z)}
  \circ q_{\psi(z)} \\[1mm]
  \delta_1(q_{\exists z (z=z \wedge \psi(z))}) &=& \displaystyle q_{\psi(z)} \vee
  \langle -1 \rangle q_{\exists z (z=z \wedge \psi(z))}
  \vee \Diamond q_{\exists z (z=z
  \wedge \psi(z))}
\\[1mm]
  \delta_1(q_{\forall z (z=z \rightarrow \psi(z))}) &=& \displaystyle q_{\psi(z)} \wedge
  [-] q_{\forall z (z=z \rightarrow \psi(z))}
  \wedge \Box q_{\forall z (z=z \rightarrow \psi(z))}
\\[1mm]
  \delta_1(q_{\exists z' \, \varphi(z,z')},\sigma) &=&
  \Diamond q_{\varphi(z,\cvar{z}')}  \vee
  q_{\varphi(z',\cvar{z})} 
 \\[1mm]
  \delta_1(q_{\forall  z' \, \varphi(z,z')},\sigma) &=&
  \Box q_{{\varphi(z,\cvar{z}')}} \wedge
   q_{{\varphi(z',\cvar{z})}} 
\end{array}
$$
$$
\begin{array}{r@{\;}c@{\;}l}
  \delta_1(q_{R z\cvar{z}'},\sigma)) &=& \left \{
    \begin{array}{ll}
\mn{true} & \text{if } Rxy \in
  \tau \\ 
\mn{false} & \text{otherwise}
    \end{array}
\right. \\[4mm]
  \delta_1(q_{R\cvar{z}z'},\sigma)) &=& \left \{
    \begin{array}{ll}
\mn{true} & \text{if } Ryx \in
  \tau \\ 
\mn{false} & \text{otherwise}
    \end{array}
\right. \\[4mm]
  \delta_1(q_{\neg Rz\cvar{z}'},\sigma)) &=& \left \{
    \begin{array}{ll}
\mn{true} & \text{if } Rxy \notin
  \tau \\ 
\mn{false} & \text{otherwise}
    \end{array}
\right. \\[4mm]
  \delta_1(q_{\neg R\cvar{z}z'},\sigma)) &=& \left \{
    \begin{array}{ll}
\mn{true} & \text{if } Ryx \notin
  \tau \\ 
\mn{false} & \text{otherwise}
    \end{array}
\right. \\[4mm]
  \delta_1(q_{\varphi(z,\cvar{z}') \circ \psi(z,\cvar{z}')},\sigma) &=& q_{\varphi(z,\cvar{z}')}
  \circ q_{\psi(z,\cvar{z}')} \\[1mm]
  \delta_1(q_{\varphi(z,\cvar{z}') \circ \psi(z)},\sigma) &=& q_{\varphi(z,\cvar{z}')}
  \circ \langle -1 \rangle q_{\psi(z)} \\[1mm]
  \delta_1(q_{\varphi(z,\cvar{z}') \circ \psi(z')},\sigma) &=& q_{\varphi(z,\cvar{z}')}
  \circ q_{\psi(z')} \\[1mm]
  \delta_1(q_{\varphi(z) \circ \psi(z')},\sigma) &=& \langle -1 \rangle q_{\varphi(z)}
  \circ q_{\psi(z')}
\end{array}
$$
where $\sigma$ ranges over $\Theta$, $z,z'$ range over $\{x,y\}$, and
$\circ$ ranges over $\{ {\wedge}, {\vee} \}$. With $\varphi(z',z)$, we
mean the result of exchanging in $\varphi(z,z')$ the variables $z$ and
$z'$, and $\overline{\varphi(z,z')}$ denotes the negation normal form
of the negation of $\varphi(z,z')$.

\medskip

We now complete the construction of the 2ATA $\Amc_2$.  It remains to
implement the obligation represented by states of the form $t^0$, that
is, the existence of $A_\bot$-delimited
GF$^2(\Sigma)$-bisimulations. Recall that such a bisimulation consists
of two relations $\sim^{A_\bot,0}_\Sigma$ and
$\sim^{A_\bot,1}_\Sigma$, each of which behaves essentially like a
GF$^2(\Sigma)$-bisimulation except in some special cases that pertain
to the $A_\bot$-marking of one of the involved structures, which in
this case is the structure $\Amf_{(T,L)}$. To deal with
$\sim^{A_\bot,0}_\Sigma$ and $\sim^{A_\bot,1}_\Sigma$, we take copies
$q^0$ and $q^1$ of every state $q$ that is of the form $t$,
$t_\downarrow$, $t_\uparrow$, $\lambda$, and $\lambda_\uparrow$, and
also copies of the above block of transitions, modified in a suitable
way to take care of the special cases. This is implemented for
$\sim^{A_\bot,0}_\Sigma$ by the following transitions:
$$\begin{array}{rcll}
  \delta_2(t^0,(\tau,M)) & = & t^0_\uparrow \wedge \Box t^0_\downarrow\wedge \displaystyle
  \bigvee_{T \mid t\approx
  T} \bigwedge_{\lambda\in
  T}( \Diamond \lambda^0 \vee \lambda^0_\uparrow )& \text{if
  }\tau_y=_\Sigma t \\
  \delta_2(t^0,(\tau,M)) & = & \mn{false} & \text{if }\tau_y\not=_\Sigma 
t 
\\[2mm]
  \delta_2(t^0_\downarrow,(\tau,M)) & = & \mn{true} & \text{if $\tau$ is
  not $\Sigma$-guarded}\\
  \delta_2(t^0_\downarrow,(\tau,M)) & = &
  \displaystyle\bigvee_{\lambda  \mid t\approx \lambda \wedge
  \tau=_\Sigma \lambda} \lambda_y^0 & \text{if
  $\tau$ is $\Sigma$-guarded}\\[2mm]
  \delta_2(t^0_\uparrow,(\tau,M)) & = & \mn{true} & \text{if $\tau$ is
  not $\Sigma$-guarded}\\
  \delta_2(t^0_\uparrow,(\tau,M)) & = &
  \displaystyle\bigvee_{\lambda  \mid t\approx \lambda \wedge
  \tau=_\Sigma \lambda^-} [-1] \lambda_y^M & \text{if
  $\tau$ is $\Sigma$-guarded}\\[2mm]
  \delta_2( \lambda^0,(\tau,M)) & = & \lambda_y^0 & \text{if }\lambda\text{
  is $\Sigma$-guarded and } \tau=_\Sigma\lambda  \\
  \delta_2( \lambda^0,(\tau,M)) & = & \mn{false} & \text{if }\lambda\text{
  is $\Sigma$-guarded and } \tau\not=_\Sigma\lambda  \\
  \delta_2( \lambda^0,(\tau,M)) & = & \lambda_y^? & \text{if $\lambda$ is
  not $\Sigma$-guarded} \\[2mm]
  \delta_2( \lambda^0_\uparrow,(\tau,M)) & = & \langle -1 \rangle (\lambda_y)_\&  & \text{if }\lambda\text{
  is $\Sigma$-guarded and } \tau=_\Sigma\lambda^-  \\
  \delta_2( \lambda^0_\uparrow,(\tau,M)) & = & \mn{false} & \text{if }\lambda\text{
  is $\Sigma$-guarded and } \tau\not=_\Sigma\lambda^-  \\
  \delta_2( \lambda^0_\uparrow,(\tau,M)) & = & \lambda_y^? & \text{if $\lambda$ is
  not $\Sigma$-guarded} \\
  \delta_2(t_\&,(\tau,M)) &=& t^M 
\end{array}$$
The transitions for $\sim^{A_\bot,1}_\Sigma$ are as follows:
$$\begin{array}{rcl@{\;}l}
  \delta_2(t^1,(\tau,1)) & = & 
  \big ( t^? \wedge t^1_\uparrow \wedge \Box t^1_\downarrow\wedge \displaystyle
  \bigvee_{T \mid t\approx
  T} \bigwedge_{\lambda\in
  T}( \Diamond \lambda^1 \vee \lambda^1_\uparrow ) \big ) & \text{if
  }\tau_y=_\Sigma t \\
&& \vee \; (t^? \wedge \langle -1 \rangle q_\bot )
\\[2mm]
  \delta_2(t^1,(\tau,0)) & = & 
  \big ( t^? \wedge t^1_\uparrow \wedge \Box t^0_\downarrow\wedge \displaystyle
  \bigvee_{T \mid t\approx
  T} \bigwedge_{\lambda\in
  T}( \Diamond \lambda^0 \vee \lambda^1_\uparrow ) \big ) & \text{if
  }\tau_y=_\Sigma t \\
&& \vee \; (t^? \wedge \langle -1 \rangle q_\bot )
\\[2mm]
  \delta_2(t^1,(\tau,M)) & = & \langle -1 \rangle q_\bot & \text{if }\tau_y\not=_\Sigma 
t \\
  \delta_2(q_\bot,(\tau,0)) & = & \mn{false} \\
  \delta_2(q_\bot,(\tau,1)) & = & \mn{true} 
\\[2mm]
  \delta_2(t^1_\downarrow,(\tau,M)) & = & \mn{true} & \text{if $\tau$ is
  not $\Sigma$-guarded}\\
  \delta_2(t^1_\downarrow,(\tau,M)) & = &
  \displaystyle\bigvee_{\lambda  \mid t\approx \lambda \wedge
  \tau=_\Sigma \lambda} \lambda_y^{1} & \text{if
  $\tau$ is $\Sigma$-guarded}\\[2mm]
  \delta_2(t^1_\uparrow,(\tau,M)) & = & \mn{true} & \text{if $\tau$ is
  not $\Sigma$-guarded}\\
  \delta_2(t^1_\uparrow,(\tau,M)) & = &
  \displaystyle\bigvee_{\lambda  \mid t\approx \lambda \wedge
  \tau=_\Sigma \lambda^-} [-1] \lambda_y^1 & \text{if
  $\tau$ is $\Sigma$-guarded}\\[2mm]
  \delta_2( \lambda^1,(\tau,M)) & = & \lambda_y^{1} & \text{if }\lambda\text{
  is $\Sigma$-guarded and } \tau=_\Sigma\lambda  \\
  \delta_2( \lambda^1,(\tau,M)) & = & \mn{false} & \text{if }\lambda\text{
  is $\Sigma$-guarded and } \tau\not=_\Sigma\lambda  \\
  \delta_2( \lambda^1,(\tau,M)) & = & \lambda_y^? & \text{if $\lambda$ is
  not $\Sigma$-guarded} \\[2mm]
  \delta_2( \lambda^1_\uparrow,(\tau,M)) & = & \langle -1 \rangle \lambda_y^1  & \text{if }\lambda\text{
  is $\Sigma$-guarded and } \tau=_\Sigma\lambda^-  \\
  \delta_2( \lambda^1_\uparrow,(\tau,M)) & = & \mn{false} & \text{if }\lambda\text{
  is $\Sigma$-guarded and } \tau\not=_\Sigma\lambda^-  \\
  \delta_2( \lambda^1_\uparrow,(\tau,M)) & = & \lambda_y^? & \text{if $\lambda$ is
  not $\Sigma$-guarded} 
\end{array}$$

\begin{lemma} $\Amc_2$ satisfies the condition from
  Lemma~\ref{lem:bisi-automata}.  \end{lemma}
\begin{proof}
  ``$\Leftarrow$''. Let $(T,L)$ be a $\Theta$-labeled tree and let
  \Bmf be a model of $\vp_2$ such that Conditions~1 and~2 of
  Theorem~\ref{thm:charsimp} are satisfied when \Amf is replaced with
  $\Amf_{(T,L)}$ (and when $A_\bot$ is the set described by the second
  component of the $L$-labels).  We argue that \Bmf can be used to
  guide a run of $\Amc_2$ on $(T,L)$ so that it is accepting.

  In this run, $\Amc_2$ starts with choosing the $0$-type $s$ realized
  by $\Bmf$.  Then, for each $\exists x \vp(x)\in s$, we guide
  $\Amc_2$ to proceed in state $t^?$, where $t$ is the $1$-type of
  some element $b \in B$ with $\Bmf \models \vp(b)$. By Condition~2 of
  Theorem~\ref{thm:charsimp}, there is a $w \in A_{(T,L)}$ such that
  $\mn{tp}_\Bmf(b)=t$ and $(\Amf_{(T,L)},w) \sim^{A_\bot}
  (\Bmf,b)$. In the search state $t^?$, we guide the run to reach $w$
  and switch to state $t^0$ there.  The automaton also sends a copy in
  state $s$ to each node $w \in A_{(T,L)}$. By Condition~1 of
  Theorem~\ref{thm:charsimp}, there is a $b \in B$ such that
  $(\Amf_{(T,L)},w) \sim_\Sigma (\Bmf,b)$. We guide the run to proceed
  in state $t$, the 1-type of $b$.

  At this point, the automaton needs to satisfy two kinds of
  obligations:
  \begin{enumerate}

  \item states $t$ true at a node $w \in A_{(T,L)}$ representing the
    obligation to verify that there is a $b \in B$ with 1-type $t$ and
    such that $(\Amf_{(T,L)},w) \sim_\Sigma (\Bmf,b)$ and

  \item states $t^0$ true at a node $w \in A_{(T,L)}$ representing the
    obligation to verify that there is a $b \in B$ with 1-type $t$ and
    such that $(\Amf_{(T,L)},w) \sim^{A_\bot}_\Sigma (\Bmf,b)$.

  \end{enumerate}
  Note that we have guided the run so that the required bisimulations
  indeed exist and therefore we can use them to further guide the run. We
  only consider Case~1 above, thus concentrating on states of the form
  $t$, $t_\downarrow$, $t_\uparrow$, $\lambda$, and
  $\lambda_\uparrow$. Suppose the automaton is in state $t$ at node
  $w$. By the way in which we guide the run, there is then a $b \in B$
  with 1-type $t$ and such that $(\Amf_{(T,L)},w) \sim_\Sigma
  (\Bmf,b)$. We guide the run to select as $T$ the set of all guarded
  2-types $\lambda$ such that $\Bmf \models (\exists y
  \lambda(x,y))(b)$. For each such $\lambda$, there must be a $b' \in
  B$ and a $v \in A_{(T,L)}$ with $\Bmf \models \lambda(b,b')$ and
  $(\Amf_{(T,L)},v) \sim_\Sigma (\Bmf,b')$ where $v$ is either the
  predecessor of $w$ or a successor of it. In the former case, we
  guide the automaton to switch to state $\lambda_\uparrow$ and in
  the latter, we guide it to execute $\Diamond \lambda$. When the
  automaton was sent in state $t_\downarrow$ to a successor $v$ of
  $w$, then there must be a $b' \in B$ such that $(\Amf_{(T,L)},v)
  \sim_\Sigma (\Bmf,b')$ and $\Bmf \models \lambda(b,b')$ for some
  guarded 2-type $\lambda$. Guide the run to choose $\lambda$.
  The decision to be taken for states $t_\uparrow$ is handled very similarly.

  \smallskip

  ``$\Rightarrow$''.  Let $(T,L)$ be a $\Theta$-labeled tree that is
  accepted by $\Amc_2$. Then there is an accepting run $(T_r,r)$ of
  $\Amc_2$ on $(T,L)$. We show how to use $(T_r,r)$ to construct a
  model \Bmf of $\vp_2$ such that Conditions~1 and~2 of
  Theorem~\ref{thm:charsimp} are satisfied when \Amf is replaced with
  $\Amf_{(T,L)}$. Along with \Bmf, we construct the following objects:
  \begin{itemize}

  \item a GF$^2(\Sigma)$-bisimulation $\sim$ between $\Amf_{(T,L)}$
    and \Bmf which witnesses that Condition~1 of
    Theorem~\ref{thm:charsimp} is satisfied,

  \item two relations $\sim^{A_\bot,0}$ and $\sim^{A_\bot,1}$ that
    form an $A_\bot$-delimited GF$^2(\Sigma)$-bisimulation 
    between $\Amf_{(T,L)}$ and \Bmf, where $A_\bot \subseteq \Amf_{(T,L)}$ is
    the subset defined by the second component of $L$, and which
    witness that Condition~2 of Theorem~\ref{thm:charsimp} is
    satisfied, and

  \item a function $\mu$ that assigns to each element of \Bmf the
    1-type that we aim to realize there.

  \end{itemize}
  Throughout the construction, we make sure that the following 
  invariants are satisfied:
  \begin{enumerate}

  \item if $(w,b) \in {\sim}$, then the label $(w,\mu(b))$ occurs in $(T_r,r)$;

  \item if $(w,b) \in {\sim^{A_\bot,i}}$, $i \in \{0,1\}$, then the
    label $(w,\mu(b)^0)$ occurs in $(T_r,r)$.

  \end{enumerate}
  The start of the construction is as follows: 
  \begin{itemize}

  \item for each label $(w,t)$ that occurs in $(T_r,r)$, introduce an
    element $b$ of $B$, add $(w,b)$ to $\sim$, and set $\mu(b)=t$;

  \item for each label $(w,t^0)$ that occurs in $(T_r,r)$, introduce
    an element $b$ of $B$, add $(w,b)$ to $\sim^{A_\bot,0}$, and set
    $\mu(b)=t$.

  \end{itemize}
  We then iteratively extend \Bmf, $\sim$, $\sim^{A_\bot,0}$,
  $\sim^{A_\bot,1}$, and $\mu$, obtaining the desired structure and
  bisimulations in the limit. In each step, process every $b \in B$
  that has not been processed in a previous round. There are three
  cases.

  \smallskip
  \noindent 
  \emph{Case~(a)}. There is a $(w,b) \in {\sim}$. By Invariant~1, we
  find a node $x \in T_r$ such that $r(x)=(w,\mu(b))$. We perform
  two steps:
  \begin{itemize}

  \item For every predecessor or successor $v$ of $w$ in~$T$ with
    $\text{at}_{\mathfrak{A}_{(T,L)}}^\Sigma(w,v)$ guarded, there
    must be a 2-type $\lambda$ such that $\mu(b) \approx \lambda$,
    $(v,\lambda_y)$ occurs as a label in $(T_r,r)$, and
    $\text{at}_{\mathfrak{A}_{(T,L)}}^{\Sigma}(w,v) =_\Sigma \lambda$.
    Extend $\Bmf$ with a new element $b'$, extend the interpretation
    of the predicates in \Bmf such that
    $\text{at}_{\Bmf}^{\Sigma}(w,v) =_\Sigma \lambda$, set
  $\mu(b')=\lambda_y$, and extend $\sim$ with $(v,b')$.

\item There must be a set $T$ of guarded 2-types such that
  $t \approx T$ and for every $\lambda \in T$, there is a predecessor
  or successor $v$ of $w$ in~$T$ such that $\mu(b) \approx \lambda$,
  $(v,\lambda_y)$ occurs as a label in $(T_r,r)$, and
  $\text{at}_{\mathfrak{A}_{(T,L)}}^{\Sigma}(w,v) =_\Sigma \lambda$.
  Extend $\Bmf$ with a new element $b'$ (for every $\lambda$), extend
  the interpretation of the predicates in \Bmf such that
  $\text{at}_{\Bmf}^{\Sigma}(w,v) =_\Sigma \lambda$, set
  $\mu(b')=\lambda_y$, and extend $\sim$ with $(v,b')$.

  \end{itemize}

  \smallskip 
  \noindent 
  \emph{Case~(b)}. There is a $(w,b) \in {\sim^{A_\bot,0}}$.  By
  Invariant~2, we find a node $x \in T_r$ such that
  $r(x)=(w,\mu(b)^0)$. We can now proceed exactly as in Case~(a)
  except that, in both subcases, we add $(v,b')$ to $\sim^{A_\bot,1}$
  if $v$ is a predecessor of $w$ and $v \in A_\bot$, and to
  $\sim^{A_\bot,0}$ otherwise.

  \smallskip 
  \noindent 
  \emph{Case~(c)}. There is a $(w,b) \in {\sim^{A_\bot,1}}$. By
  Invariant~2, we find a node $x \in T_r$ such that
  $r(x)=(w,\mu(b)^1)$. If the predecessor of $w$ is not in $A_\bot$,
  then we again proceed as in Case~(a) except that, in both subcases,
  we add $(v,b')$ to $\sim^{A_\bot,0}$ if $v$ is a sucessor of $w$
  and $w \in A_\bot$, and to $\sim^{A_\bot,1}$ otherwise.
  If the predecessor of $w$ is in $A_\bot$, then we we also proceed as
  in Case~(a), but do not add $(v,b')$ to any of the constructed
  bisimulations.

  \smallskip 
  \noindent 
  \emph{Case~(d)}. None of the above cases applies. Then we proceed
  as in Case~(a), again not adding $(v,b')$ to any of the constructed
  bisimulations.

  \medskip
  \noindent
  It can be verified that, as intended the structure \Bmf obtained in
  the limit is a model of $\vp_2$, that the relation $\sim$ is a
  GF$^2(\Sigma)$-bisimulation, and that
  $\sim^{A_\bot,0},\sim^{A_\bot,1}$ form an $A_\bot$-delimited
  GF$^2(\Sigma)$-bisimulation.
\end{proof}
Recall that we define the overall 2ATA \Amc so that it accepts
$L(\Amc_1) \cap \overline{L(\Amc_2)} \cap L(\Amc_3)$. Using
Theorem~\ref{thm:charsimp}, it can be verified that, as intended,
$\vp_1 \models_{\text{GF}^2(\Sigma)} \vp_2$ iff $L(\Amc)=\emptyset$.
Note that for the ``only if'' direction, we have to show that
$L(\Amc) \neq \emptyset$ implies that there is a \emph{regular} forest
model of $\vp_1$ that satisfies the negation of the conditions in
Theorem~\ref{thm:charsimp}. As is the case for other kinds of tree
automata, also for the 2ATA \Amc it can be shown that
$L(\Amc) \neq \emptyset$ implies that \Amc accepts a regular
$\Theta$-labeled tree $(T,L)$. The corresponding structure
$\Amf_{(T,L)}$ must then also be regular.

\bigskip

We show that $\Sigma$-entailment, $\Sigma$-inseparability, and
conservative extensions in GF$^2$ are 2\ExpTime-hard.  The proof is
by reduction of the word problem for exponentially space bounded
alternating Turing machines (ATMs). An \emph{ATM} is of the form $M =
(Q,\Theta,\Gamma,q_0,\Delta)$. The set of \emph{states} $Q = Q_\exists
\uplus Q_\forall \uplus \{q_a\} \uplus \{q_r\}$ consists of
\emph{existential states} from $Q_\exists$, \emph{universal states}
from $Q_\forall$, an \emph{accepting state} $q_a$, and a
\emph{rejecting state} $q_r$; $\Theta$ is the \emph{input alphabet}
and $\Gamma \supset \Theta$ the \emph{work alphabet} that contains a
\emph{blank symbol} $\square \notin \Theta$; $q_0 \in Q_\exists$ is
the \emph{starting} state; and the \emph{transition relation} $\Delta$
is of the form $ \Delta \; \subseteq \; Q \times \Gamma \times Q
\times \Gamma \times \{ L,R \}.  $
We write $\Delta(q,a)$ for $\{ (q',b,M) \mid (q,a,q',b,M) \in \Delta
\}$ and assume that $\Delta(q,b) = \emptyset$ for all $q \in \{
q_a,q_r \}$ and $b \in \Gamma$.

A \emph{configuration} of an ATM is a word $wqw'$ with $w,w' \in
\Gamma^*$ and $q \in Q$. The intended meaning is that the one-side
infinite tape contains the word $ww'$ with only blanks behind it, the
machine is in state $q$, and the head is on the symbol just after
$w$. The \emph{successor configurations} of a configuration $wqw'$ are
defined in the usual way in terms of the transition
relation~$\Delta$. A \emph{halting configuration} (resp.\
\emph{accepting configuration}) is of the form $wqw'$ with $q \in \{
q_a, q_r \}$ (resp.\ $q=q_a$).

A \emph{computation tree} of an ATM $M$ on input $w$ is a tree whose
nodes are labeled with configurations of $M$ on $w$, such that the
descendants of any non-leaf labeled by a universal (resp.\
existential) configuration include all (resp. one) of the successor
configurations of that configuration. A computation tree is
\emph{accepting} if the root is labeled with the \emph{initial
  configuration} $q_0w$ for $w$ and all leaves with accepting
configurations. An ATM $M$ accepts input $w$ if there is a computation
tree of $M$ on $w$.

Take an exponentially space bounded ATM $M$ whose word problem is
{\sc 2ExpTime}-hard~\cite{Chandra-et-al-81}. We may w.l.o.g.\ assume
that the length of every computation path of $M$ on $w \in \Theta^n$
is bounded by $2^{2^{n}}$. We can also assume that for each $q \in
Q_\forall \cup Q_\exists$ and each $a \in \Gamma$, the set
$\Delta(q,a)$ has exactly two elements. We assume that these elements
are ordered, i.e., $\Delta(q,a)$ is an ordered pair
$((q_L,b_L,M_L),(q_R,b_R,M_R))$.
Furthermore, we assume that $M$ never attempts to move left on
the left-most tape cell.

Let $w=a_0\cdots a_{n-1} \in \Theta^*$ be an input to $M$.  In the
following, we construct GF$^2$ sentences $\varphi_1$ and $\varphi_2$
such that $\vp_1 \wedge \vp_2$ is a conservative extension of $\vp_1$
if and only if $M$ does not accept $w$.  Informally, the main idea is
to construct $\vp_1$ and $\vp_2$ such that models of sentences that
witness non-conservativity describe an accepting computation tree of
$M$ on $w$. In such models, each domain element represents a tape cell
of a configuration of $M$, the binary predicate $N$ indicates moving
to the next tape cell in the same configuration, and the binary
predicates $L$ and $R$ indicate moving to left and right successor
configurations in accepting configuration trees. Thus, each node of
the computation tree (that is, each configuration) is spread out over
a sequence of nodes in the model. We actually assume that every
non-halting configuration has two successor configurations, also when
its state is existential. This can of course easily be achieved by
duplicating subtrees in computation trees. The following predicates
are used in $\vp_1$:
\begin{itemize}

\item a unary predicate $P$ to mark the root of computation
  trees;

\item binary predicates $N$, $R$, $L$, as explained above;

\item unary predicates $C_0,\dots,C_{n-1}$ that represent the bits
  of a binary counter which identifies tape positions;

\item a unary predicate $F$ that marks the topmost configuration
  in the configuration tree;

\item unary predicates $S_a$, $a \in \Gamma$, to represent the tape 
  content of cells that are not under the head;

\item unary predicates $S_{q,a}$, $q \in Q$ and $a \in \Gamma$, to
  represent the state of a configuration, the head position, and the
  tape content of the cell that is under the head;

\item unary predicates $S^p_a$ and $S^p_{q,a}$, with the ranges of $q$
  and $a$ as above, to represent the same information, but for the
  previous configuration in the tree instead of for the current one;

\item unary predicates $Y_{L,q,a,M}$ and $Y_{R,q,a,M}$, $q \in Q$, $a
  \in \Gamma$, $M \in \{L,R\}$, to record the transition to be
  executed in the subsequent configurations;

\item unary predicates $Y_{q,a,M}$, with the ranges of $q,a,M$ as
  above, to record the transition executed to reach the current
  configuration.

\end{itemize}
The sentence $\vp_2$ uses some additional unary predicates,
including $C'_0,\dots,C'_{n-1}$ to implement another counter whose
purpose is explained below.

\begin{figure}[t!]
  \begin{boxedminipage}[t]{\columnwidth}
\vspace*{-4mm}
  \begin{center}
\begin{align}
&\forall x ( Px \rightarrow \varphi_{C=0}(x) )\\
&\forall x \big ( \varphi_{C < 2^n-1}(x) \rightarrow \big ( \exists y Nxy \wedge
\forall y (Nxy \rightarrow \varphi_{C{+}{+}}(x,y) \big ) \big)\\
&\forall x \big ( \varphi_{C = 2^n-1}(x) \rightarrow (\exists y Lxy 
\wedge \forall y (Lxy \rightarrow \varphi_{C=0}(y)) \wedge \exists y Rxy
\wedge \forall y (Rxy \rightarrow \varphi_{C=0}(y))) \big )\\
&\forall x \big ((Px \rightarrow Fx ) \wedge \forall y
 (Nxy \rightarrow Fy) \big)\\[3mm]
%
&\forall x \bigvee_{\alpha \in \Gamma \cup (Q \times \Gamma)} \big (
S_\alpha x \wedge \bigwedge_{\beta \in (\Gamma \cup (Q \times 
  \Gamma))\setminus \{ \alpha \}} \neg S_\beta x \big )\\
&\forall x \bigvee_{\alpha \in \Gamma \cup (Q \times \Gamma)} \big (
S^p_\alpha x \wedge \bigwedge_{\beta \in (\Gamma \cup (Q \times 
  \Gamma))\setminus \{ \alpha \}} \neg S^p_\beta x \big )\\[3mm]
&\forall x \big ((Fx \wedge \varphi_{C=0}(x))\rightarrow S_{q_0,a_0}x\big)\\
&\forall x \big ( (Fx \wedge \varphi_{C=i}(x)) \rightarrow S_{a_i}x \big
  ) \qquad \text{ for } 1\leq i < n\\
&\forall x \big ((Fx \wedge\vp_{C\geq n}(x)) \rightarrow S_\Box x \big)
\\[3mm]
& \forall x \big ( S_{q,a} x \rightarrow (Y_{L,T_L}x \wedge
Y_{R,T_R} x)\big ) \hspace*{2.65cm}\text{ if } \Delta(q,a)=(T_L,T_R), \ q \in Q_\forall\\
& \forall x \big ( S_{q,a} x \rightarrow ((Y_{L,T_L}x \wedge
Y_{R,T_L}x) \vee (Y_{L,T_R}x \wedge
Y_{R,T_R}x)) \big ) \\
& \hspace*{7.5cm} \text{ if } \Delta(q,a)=(T_L,T_R), \ q \in Q_\exists
\nonumber\\
&\forall x \big ( Y_{P,T}x \rightarrow \forall y (Nxy \rightarrow Y_{P,T}y)  \big ) \\
&\forall x \big ( Y_{P,T}x \rightarrow \forall y (Pxy \rightarrow
Y_{T}y) \big ) \\
&\forall x \big ( Y_{T}x \rightarrow \forall y (Nxy \rightarrow Y_{T}y)  \big ) \\[3mm]
&\forall x \big ((Y_{q,a,M}x \wedge S^p_{q',b}x) \rightarrow S_ax \big) \\
&\forall x \big ((Y_{q,a,L}x \wedge S^p_bx \wedge 
\exists y (Nxy \wedge S^p_{q',a'}y)) \rightarrow S_{q,b}x \big ) \\
&\forall x \big ((Y_{q,a,R}x \wedge S^p_bx \wedge 
\exists y (Nxy \wedge S^p_{q',a'}y)) \rightarrow S_{b}x \big ) \\
&\forall x \big ((Y_{q,a,R}x \wedge S^p_bx \wedge 
\exists y (Nyx \wedge S^p_{q',a'}y)) \rightarrow S_{q,b}x \big ) \\
&\forall x \big ((Y_{q,a,L}x \wedge S^p_bx \wedge 
\exists y (Nyx \wedge S^p_{q',a'}y)) \rightarrow S_{b}x \big ) \\
&\forall x \big ((\exists y (Nxy \wedge S^p_{b}y)) \wedge S^p_ax
\wedge \exists y (Nyx \wedge S^p_{b'}y))\rightarrow S_{a}x \big )  \\[3mm]
&\forall x \neg S_{q_r,a}x 
\end{align}
  \end{center}
  \end{boxedminipage}
  \caption{The conjuncts of the sentence $\varphi_1$.}
  \label{fig:tbox1}
\end{figure}

The sentences $\vp_1$ and $\vp_2$ are shown in Figures~\ref{fig:tbox1}
and~\ref{fig:compltbox2}, respectively, where $q$ and $q'$ range over~$Q$, $a,b,b'$
over $\Gamma$, $M$ and $P$ over $\{L,R\}$,
$T,T_L,T_R$ over $Q \times \Gamma \times \{L,R\}$, and $\alpha$ over
$\Gamma \cup (Q \times \Gamma)$. The formula $\vp_{C=i}(x)$, which is
easily worked out in detail, expresses that the value of the binary
counter implemented by $C_0,\dots,C_{n-1}$ has value exactly $i$ at
$x$, and likewise for $\vp_{C<i}(x)$ and $\vp_{C\geq i}(x)$, and for
the primed versions in $\vp_2$ which refer to the counter implemented
by $C'_0,\dots,C'_{n-1}$. The formula $\vp_{C{+}{+}}(x,y)$ expresses
that the counter value at $y$ is obtained from the counter value at
$x$ by incrementation modulo $2^n$.  Again, we omit the details.

Let us walk through $\vp_1$ and $\vp_2$ and give some intuition of
what the various conjuncts are good for. In $\vp_1$, Lines~(1) to~(4)
ensure that at an element that satisfies $P$, there is an infinite
tree of the expected pattern: first $2^n-1$ $N$-edges without
branching, then a binary branching of an $L$-edge and an $R$-edge,
then $2^n-1$ $N$-edges without branching, and so on, ad infinitum. Of
course, a computation tree will be represented using only a finite
initial piece of this infinite tree. These conjuncts also set up the
counter $C$ so that it identifies the position of tape cells and the
marker $F$ so that it identifies the topmost configuration in the
tree. Line~(5) says that every cell is labeled with exactly one symbol
and that the state is unique (locally to one cell; there is no need to
express the same globally for the entire configuration), and Line~(6)
says the same for the representation of the previous configuration.
Lines~(7) to~(9) make sure that the topmost configuration in the
infinite tree is the initial configuration of $M$ on input $w$.
Lines~(10) and~(11) choose transitions to execute and Lines~(12)
to~(14) propagate this choice down to the subsequent configurations.
Assume that the predicates of the $S^p_a$ and $S^p_{q,a}$ indeed
represent the previous configuration, Lines~(15) to~(20) then
implement the chosen transitions. Line~(21) says that we do never see
a rejecting halting configuration.

%
\begin{figure}[t!]
  \begin{boxedminipage}[t]{\columnwidth}
\vspace*{-4mm}
  \begin{center}
\begin{align}
&\exists x Px \rightarrow \exists x Dx\\[3mm]
&\forall x  \big (Dx \rightarrow (Mx \wedge \vp_{C'=0}(x)) \big )\\
%
&\forall x \big( (Dx \wedge S_\alpha x) \rightarrow Z_\alpha x \big) \\[3mm]
&\forall x \big ((Mx \wedge \vp_{C<2^n-1}(x)  \wedge \vp_{C' < 2^n-1}(x) \wedge
  Z_{\alpha}x)
  \rightarrow \exists y (Nxy \wedge My \wedge Z_\alpha y \wedge \vp_{C'{+}{+}}(x,y)) \big ) \\
&\forall x \big ((Mx \wedge \vp_{C=2^n-1}(x) \wedge \vp_{C' < 2^n-1}(x) \wedge
  Z_{\alpha}x)
  \rightarrow \\
& \hspace*{1.5cm} \exists y (Lxy \wedge My \wedge Z_{\alpha} y \wedge \vp_{C'{+}{+}}(x,y)) \vee
\exists y
  (Rxy ~\wedge
 My \wedge Z_{\alpha}y \wedge \vp_{C'{+}{+}}(x,y)) \big ) \nonumber\\
&\forall x \big ((Mx \wedge \vp_{C'=2^n-1}(x) \wedge Z_\alpha x) \rightarrow
  \exists y (Nxy \wedge \neg S^p_ \alpha x) \big ) 
\end{align}
  \end{center}
  \end{boxedminipage}
  \caption{The conjuncts of the sentence $\varphi_2$.}
  \label{fig:compltbox2}
\end{figure}

Now for $\vp_2$. Essentially, we want to achieve that a sentence is a
witness for non-conservativity if and only if it expresses that its
models contain (a representation of) an accepting computation tree of
$M$ on $w$ whose root is labeled with $P$.  This is achieved by
designing $\vp_2$ so that, whenever a tree model of $\vp_1$ contains
only instances of $P$ that are not the root of such a computation
tree, then this model can be extended to a model of $\vp_2$ by
assigning an interpretation to the additional predicates in
$\vp_2$. Note that $\vp_1$ already enforces that, below any instance
of $P$, there is a tree that satisfies almost all of the required
conditions of an accepting computation tree. In fact, the only way in
which that tree cannot be an accepting configuration tree is that the
predicates $S^p_a$ and $S^p_{q,a}$ do not behave in the expected way,
that is, there is a configuration and a cell in this configuration
that is labeled with $S_\alpha$, $\alpha \in \Gamma \cup (Q \times
\Gamma)$, and in one of the two subsequent configurations the same
cell is not labeled with $S^p_\alpha$. We thus design $\vp_2$ so that
it can be made true whenever the model contains such a defect. In
Line~(22), we select the place where the defect is. Line~(23)
ensures that the counter $C'$ starts counting with value zero
at that place, and that the marker $M$ is set there,
too. Line~(24) memorizes the content $\alpha$ of the cell in the upper
configuration involved in the defect. Lines~(25) to~(27) propagate
downwards the memorized content and make sure that, at the
corresponding cell of at least one subsequent configuration (which is
identified using the counter $C'$), we do not find $S^p_\alpha$.
\begin{lemma}\label{lem:equivalenceATM}
~\\[-4mm]
  \begin{enumerate}

  \item
    If $M$ accepts $w$, then $\vp_{1} \wedge \vp_{2}$ is not a
    GF$^2$-conservative extension of $\vp_{1}$.

  \item If there exists a $\sig(\vp_{1})$-structure that satisfies $\vp_{1}$
    and cannot be extended to a model of $\vp_{1}\wedge \vp_{2}$,
    then $M$ does not accept $w$.
  \end{enumerate}
\end{lemma}
\begin{proof}(sketch) For Point~1 assume that $M$ accept $w$. Then there is an accepting
  computation tree of $M$ on $w$. Let $\Sigma=\sig(\vp_{1})$. We can
  find a GF$^2(\Sigma)$-sentence 
  $\psi_{1}$ which expresses that the model contains a (homomorphic image of a) finite
  tree which represents this configuration tree and whose root is
  labeled with $P$. We can also find a GF$^2(\Sigma)$-sentence
  $\psi_2$ which expresses that nowhere in the model there is a defect
  situation. It can be verified that $\psi_1 \wedge \psi_2$ is
  satisfiable w.r.t.\ $\vp_1$, but not w.r.t.\ $\vp_2$ because $\vp_2$
  requires the existence of a defect situation whenever the extension
  of $P$ is non-empty.

\medskip For Point~2 assume that $\Amf$ is a $\sig(\vp_{1})$-structure
that satisfies $\vp_{1}$.  If $P^\Amf = \emptyset$, then the desired
model \Bmf is obtained from \Amf by interpreting all predicates in
$\mn{sig}(\vp_2) \setminus \mn{sig}(\vp_1)$ as empty.  Otherwise, take
some $a \in P^\Amf$. We can follow the existential quantifiers in
$\vp_1$ to identify a homomorphic image of an infinite tree in \Amf
with root $a$ whose edges follow the expected pattern and that is
labeled in the expected way by the counter $C$.
  Since \Amf is a model of $\vp_1$, an initial piece of the identified
  tree represents an accepting computation tree of $M$ on $w$ provided
  that the predicates $S^p_\alpha$ behave as expected, that is, if
  there is no defect of the form described above. Since $M$ does not
  accept $w$, there must thus be such a defect, that is a path of
  length $2^n$ that links a cell of a configuration with the
  corresponding cell of a subsequent configuration such that the
  former is labeled with $S_\alpha$, but the latter is not labeled
  with $S^p_\alpha$. All the elements of the path (with the possible
  exception of the start point and the end point) are labeled with a
  different value of the counter $C$ and must thus be distinct.
  Consequently, we can interpret the counter $C'$ and the other
  symbols in $\vp_2$ to extend \Amf to a model of $\vp_2$, as desired.
\end{proof}
It follows directly from Lemma~\ref{lem:equivalenceATM} that
$\Sigma$-entailment, $\Sigma$-inseparability, and conservative
extensions in GF$^2$, are 2\ExpTime-hard.

\section{2ATAs and Their Emptiness Problem}

The aim of this section is to show that the emptiness problem for
2ATAs can be solved in time exponential in the number of states. For
proving this, we reduce it to the emptiness problem of the standard
two-way alternating tree automata over trees of fixed
outdegree~\cite{Vardi98}.

We start with making precise the semantics of 2ATAs. Let
$\Amc=(Q,\Theta,q_0,\delta,\Omega)$ be a 2ATA and $(T,L)$ a
$\Theta$-labeled tree. A {\em run for \Amc on $(T,L)$} is a $T\times
Q$-labeled tree $(T_r,r)$ such that:
\begin{itemize}

  \item $\varepsilon\in T_r$ and $r(\varepsilon)=(\varepsilon,q_0)$;

  \item For all $y\in T_r$ with $r(y)=(x,q)$ and $\delta(q,L(x))=\vp$,
    there is an assignment $v$ of truth values to the transitions in $\vp$
    such that $v$ satisfies $\vp$ and:
    \begin{itemize}

      \item if $v(p)=1$, then $r(y')=(x,p)$ for some successor
	$y'$ of $y$ in $T_r$;

      \item if $v(\langle - \rangle p)=1$, then $x \neq \varepsilon$
        and there is a successor $y'$ of $y$ in $T_r$ with
        $r(y')=(x\cdot -1,p)$;

      \item if $v([-] p)=1$, then $x=\varepsilon$ or there is a
        successor $y'$ of $y$ in $T_r$ such that 
	$r(y')=(x\cdot -1,p)$;
        
      \item if $v(\Diamond p)=1$, then there is a successor $x'$ of
        $x$ in $T$ and a successor $y'$ of $y$ in $T_r$ such that
        $r(y')=(x',p)$;

      \item if $v(\Box p)=1$, then for every successor $x'$ of
        $x$ in $T$, there is a successor $y'$ of $y$ in $T_r$ such that
        $r(y')=(x',p)$.

    \end{itemize}

\end{itemize}
Let $\gamma=i_0i_1\cdots$ be an infinite path in $T_r$ and denote, for
all $j\geq 0$, with $q_j$ the state such that $r(i_0\cdots
i_j)=(x,q_j)$. The path $\gamma$ is {\em accepting} if the largest
number $m$ such that $\Omega(q_j)=m$ for infinitely many $j$ is even.
A run $(T_r,r)$ is accepting, if all infinite paths in $T_r$ are
accepting. Finally, a tree is accepted if there is some accepting run
for it.

\smallskip

We next introduce strategy trees similar to~\cite[Section 4]{Vardi98}. A
{\em strategy tree for a 2ATA $\Amc$} is a tree $(T,\tau)$ where
$\tau$ labels every node in $T$ with a subset $\tau(x)\subseteq
2^{Q\times (\mathbb{N}\cup\{-1\}) \times Q}$, that is, with a graph with
nodes from $Q$ and edges labeled with natural numbers or $-1$.
Intuitively, $(q,i,p)\in\tau(x)$ expresses that, if we reached
node $x$ in state $q$, then we should send a copy of the automaton in
state $p$ to $x\cdot i$. For each label $\zeta$, we define
$\mn{state}(\zeta)=\{q\mid (q,i,q')\in \zeta\}$, that is, the set of
sources in the graph $\zeta$. A strategy tree is {\em on an input tree
$(T',L)$} if $T=T'$, $q_0\in \mn{state}(\tau(\varepsilon))$, and for
every $x\in T$, the following conditions are satisfied:
\begin{enumerate}

  \item if $(q,i,p)\in \tau(x)$, then $x\cdot i\in T$; 

  \item if $(q,i,p)\in \tau(x)$, then $p\in \mn{state}(\tau(x\cdot i))$;

  \item if $q\in \mn{state}(\tau(x))$, then the truth
    assignment $v_{q,x}$ defined below satisfies $\delta(q,L(x))$:
    \begin{enumerate}

  \item $v_{q,x}(p)=1$ iff $(q,0,p)\in \tau(x)$;

  \item $v_{q,x}(\langle-\rangle p) = 1$ iff $(q,-1,p)\in \tau(x)$;

  \item $v_{q,x}([-] p) = 1$ iff $x=\varepsilon$ or $(q,-1,p)\in \tau(x)$;

  \item $v_{q,x}(\Diamond p) = 1$ iff there is some $i$ with
    $(q,i,p)\in \tau(x)$;

  \item $v_{q,x}(\Box p) = 1$ iff $(q,i,p)\in\tau(x)$, for all $x\cdot
    i\in T$;

\end{enumerate}

\end{enumerate}
A {\em path $\beta$} in a strategy tree $(T,\tau)$ is a sequence
$\beta=(u_1,q_1)(u_2,q_2)\cdots$ of pairs from $T\times Q$ such
that for all $\ell>0$, there is some $i$ such that
$(q_\ell,i,q_{\ell+1})\in \tau(u_\ell)$ and $u_{\ell+1}=u_\ell\cdot
i$. Thus,
$\beta$ is obtained by following moves prescribed by the strategy tree.
We say that $\beta$ is \emph{accepting} if the largest number $m$ such that
$\Omega(q_i)=m$, for infinitely many $i$, is even. A strategy tree
$(T,\tau)$ is \emph{accepting} if all infinite paths in $(T,\tau)$ are accepting.

\begin{lemma}
  A 2ATA accepts a $\Theta$-labeled tree $(T,L)$ iff there is an
  accepting strategy tree for \Amc on $(T,L)$.
\end{lemma}
\begin{proof}
  The ``if''-direction is immediate: just read off an
  accepting run from the accepting strategy tree.

  For the ``only if''-direction, we observe that acceptance of an
  input tree can be defined in terms of a parity game between Player 1
  (trying to show that the input is accepted) and Player 2 (trying to
  challenge that). The initial configuration is $(\varepsilon,q_0)$
  and Player 1 begins. Consider a configuration $(x,q)$. Player 1
  chooses a satisfying truth assignment $v$ of $\delta(q,L(x))$.
  Player 2 chooses a transition $\alpha$ with
  $v(\alpha)=1$ and the next configuration is determined as follows:
  \begin{itemize}

    \item if $\alpha=p$, then the next configuration is $(x,p)$,

    \item if $\alpha=\langle-\rangle p$, then the next configuration is
      $(x\cdot -1,p)$ unless $x=\varepsilon$ in which case Player~1
      loses immediately.

    \item if $\alpha=[-] p$, then the next configuration is $(x\cdot
      -1,p)$ unless $x=\varepsilon$ in which case Player~2 loses
      immediately;

    \item if $\alpha=\Diamond p$, then Player~1 chooses some $i$ with
      $x\cdot i\in T$ (and loses if no such $i$ exists) and the next
      configuration is $(x\cdot i,p)$;

    \item if $\alpha=\Box p$, then Player~2 chooses some $i$ with
      $x\cdot i\in T$ (and loses if no such $i$ exists) and the next
      configuration is $(x\cdot i,p)$.

  \end{itemize}
  Player 1 wins an infinite play $(x_0,q_0)(x_1,q_1)\cdots$ if the
  largest number $m$ such that $\Omega(q_i)=m$, for infinitely many
  $i$, is even. It is not difficult to see that Player~1 has a winning
  strategy on an input tree iff \Amc accepts the input tree.
  Observe now that the defined game is a parity game and thus Player 1
  has a winning strategy iff she is has a {\em memoryless} winning
  strategy~\cite{EmersonJ91}. It remains to observe that a memoryless
  winning strategy is nothing else than an accepting strategy tree.
\end{proof}

Based on the previuos lemma, we show that, if $L(\Amc)$ is not empty,
then it contains a tree of small outdegree.
\begin{lemma} \label{lem:bounded-outdegree}
  If $L(\Amc)\neq\emptyset$, then there is a $(T,L)\in L(\Amc)$ 
  such that the outdegree of $T$ is bounded by the
  number of states in $\Amc$.
\end{lemma}
\begin{proof}
  Let $(T,L)$ be a $\Theta$-labeled tree and $\tau$ an accepting
  strategy tree on $T$. We construct a tree $T'\subseteq T$ and an
  accepting strategy tree $\tau'$ on $(T',L')$ where $L'$ is the
  restriction of $L$ to~$T'$. Start with $T' = \{\varepsilon\}$ and
  $\tau'$ the empty mapping.  Then exhaustively repeat the following
  step. Select an $x \in T'$ with $\tau'(x)$ undefined, in a fair
  way. Then construct $\tau'(x)$ as follows:
\begin{enumerate}

\item for every $(q,i,p)\in \tau(x)$ with $i\in\{-1,0\}$, include
  $(q,i,p)$ in $\tau'(x)$;

\item for every $p\in Q$, choose an $i$ such that $(q,i,p)\in
    \tau(x)$ for some $q$, if existant. Then add $x\cdot i$ to $T'$
    and include $(q',i,p)$ in $\tau'(x)$ for all $(q',j,p)\in
    \tau(x)$;
    
  \item further include $(q,i,p)$ in $\tau'(x)$ whenever $x \cdot i
    \in T'$ and $(q,j,p) \in \tau(x)$ for all $j$ with $x \cdot j \in
    T$.

\end{enumerate}
Clearly, $T'$ has the desired outdegree.  It remains to show that
$\tau'$ is an accepting strategy tree on~$(T',L')$. Observe
that the following properties hold for all $x\in T'$, and $p,q\in Q$:
\begin{itemize}
  \item[(i)] $(q,i,p)\in \tau(x)$ iff $(q,i,p)\in \tau'(x)$, for
    $i\in\{-1,0\}$;

  \item[(ii)] $(q,i,p)\in \tau(x)$ for some $i>0$ with $x\cdot i\in T$
    iff $(q,j,p)\in \tau'(x)$ for some $j>0$ with $x\cdot j\in T$.

\end{itemize}
Observe that we have $q_0\in\mn{state}(\tau'(\varepsilon))$, by
Points~(i) and~(ii) and since $q_0\in\mn{state}(\tau(\varepsilon))$.
It can be verified that Conditions~1 and~2 of a strategy tree being on
an input tree are satisfied due to the construction of $T'$ and
$\tau'$. For Condition~3, take any $x\in T'$ and $q\in
\mn{state}(\tau'(x))$. As $q\in\mn{state}(\tau(x))$, we know that the
truth assignment $v_{q,x}$ defined for $\tau$ satisfies
$\delta(q,V(x))$. Let $v'_{q,x}$ be the truth assignment for
$\tau',q,x$. It suffices to show that, for all transitions $\alpha$,
$v_{q,x}(\alpha)=1$ implies $v'_{q,x}(\alpha)=1$. By Point~(i), this
is the case for transitions of the form $p,\langle-\rangle
p,[-]p$. For $\alpha=\Diamond p$, we know that there is some $i,p$
with $(q,i,p)\in\tau(x)$. By Point~(ii), we know that $(q,i',p)\in
\tau'(x)$ for some $i'$ with $x\cdot i'\in T'$, and thus,
$v_{q,x}'(\alpha)=1$. For $\alpha=\Box p$, we know that $(q,i,p) \in
\tau(x)$ for all $i$ with $x \cdot i \in T$. By construction if
$\tau'$, it follows that $(q,i,p) \in \tau(x)$ for all $i$ with $x
\cdot i \in T'$, as required.

We finally argue that $\tau'$ is also accepting. Let
$\beta=(u_1,q_1)(u_2,q_2)\cdots$ be an infinite path in $(T',\tau')$.
We construct an infinite path
$\beta'=(u_1',q_1)(u_2',q_2)(u_3',q_3)\cdots$ in $(T,\tau)$
as follows:
\begin{itemize}

  \item $u_1'=u_1$;

  \item if $u_{i+1}=u_i\cdot \ell$ with $\ell\in\{-1,0\}$, then
    $u_{i+1}' = u_{i}'\cdot \ell$.

  \item if $u_{i+1}=u_i\cdot \ell$ for some $\ell\geq 0$ with
    $(q_i,\ell,q_{i+1})\in \tau'(x)$, then, by Point~(ii), there is
    some $\ell'$ with $(q_i,\ell',q_{i+1})\in \tau(x)$ and
    $x\cdot\ell'\in T'$. Set $u_{i+1}=u_i\cdot\ell'$.

\end{itemize}
Since every infinite path in $(T,\tau)$ is accepting, so is $\beta'$,
and thus $\beta$.
\end{proof}

We are now reduce to reduce the emptiness problem of 2ATAs to the
emptiness of alternating automata running on trees of fixed
outdegree~\cite{Vardi98}, which we recall here. A tree $T$ is
\emph{$k$-ary} if every node has exactly $k$ successors. A \emph{two-way
  alternating tree automaton over $k$-ary trees (2ATA$^k$)} that 
are $\Theta$-labeled is a tuple $\Amc=(Q, \Theta, q_0, \delta, \Omega)$ where
$Q$ is a finite set of \emph{states}, $\Theta$ is the \emph{input
  alphabet}, $q_0 \in Q$ is an \emph{initial state}, $\delta$ is the
\emph{transition function}, and $\Omega:Q\to\Nbbm$ is a \emph{priority
  function}. The transition function maps a state $q$ and some input
letter $\theta$ to a \emph{transition condition} $\delta(q,\theta)$,
which is a positive Boolean formula over the truth constants
$\mn{true}$, $\mn{false}$, and transitions of the form $(i,q)\in
[k]\times Q$ where $[k]=\{-1,0,\ldots,k\}$.  A \emph{run} of $\Amc$ on
a $\Theta$-labeled tree $(T,L)$ is a $T \times Q$-labeled tree $(T_r,
r)$ such that
\begin{enumerate}

  \item $r(\varepsilon)=(\varepsilon,q_0)$; 

  \item for all $x \in T_r, r(x) =(w,q)$, and $\delta(q,\tau(w))=
    \vp$, there is a (possibly empty) set $\mathcal S = \{(m_1,q_1), \ldots,
    (m_n,q_n)\} \subseteq [k] \times Q$ such that $\mathcal S$
    satisfies~$\vp$ and for $1 \leq i \leq n$, we have $x\cdot i
    \in T_r$, $w \cdot m_i$ is defined,
    and $\tau_r(x \cdot i)= (w\cdot m_i, q_i)$.  

\end{enumerate}
Accepting runs and accepted trees are defined as for 2ATAs. The
emptiness problem for 2ATA$^k$s can be solved in time exponential in the
number of states~\cite{Vardi98}. 
 
\begin{theorem}\label{thm:cbata-emptiness}
  The emptiness problem for 2ATAs can be solved in time exponential in
  the number of states.
\end{theorem}

\begin{proof}
  Let $\Amc=(Q,\Theta,q_0,\delta,\Omega)$ be a 2ATA with $n$ states.
  We transform \Amc into a 2ATA$^n$ $\Amc'=(Q',
  \Theta',q_0',\delta',\Omega)$, running over $n$-ary
  $\Theta'$-labeled trees, where $Q'= Q\uplus\{q_0',q',q_r\}$ and
  $\Theta'=\Theta\times\{0,1\}$. The extended alphabet and the extra
  states $q_0',q',q_r$ are used to simulate transitions of the form
  $[-]p$. We make sure that the additional component labels the root
  node with 1 and all other nodes with 0, and based on this use $q_r$
  to check whether we are at the root of the input tree.
  
  Formally, we proceed as follows. For all $q\in Q$, $\theta\in
  \Theta$, and $b\in \{0,1\}$ obtain $\delta'(q,(\theta,b))$ from
  $\delta(q,\theta)$ by replacing $q$ with $(0,q)$, $\langle-\rangle
  q$ with $(-1,q)$, $[-] q$ with $(0,q_r)\vee (-1,q)$, $\Diamond q$
  with $\bigvee_{i=1}^n(i,q)$, and $\Box q$ with
  $\bigwedge_{i=1}^n(i,q)$.  To enforce the intended labeling in the
  second component and the correct
  behaviour for $q_r$, we set:
  $$ \begin{array}[h]{rcll}
    \delta'(q_0',(\theta,b)) &=&  \left\{\begin{array}{ll} \mn{false}
      & \text{if $b=0$} \\
      (0,q_0)\wedge \bigwedge_{i=1}^k (i,q') & \text{otherwise}
    \end{array}\right.\\[4mm]
    \delta'(q',(\theta,b)) &=&  \left\{\begin{array}{ll}
      \bigwedge_{i=1}^k (i,q') & \text{if $b=0$} \\
      \mn{false} & \text{otherwise}
    \end{array}\right.\\[4mm]
    \delta'(q_r,(\theta,b)) &=&  \left\{\begin{array}{ll} \mn{true}
      & \text{if $b=1$} \\
    \mn{false}& \text{otherwise}
    \end{array}\right.
  \end{array}$$
 Using Lemma~\ref{lem:bounded-outdegree} it is straightforward to verify
  that $L(\Amc)\neq\emptyset$ iff $L(\Amc')\neq\emptyset$. Since the
  translation can be done in polynomial time and the emptiness problem
  for 2ATA$^k$s is in \ExpTime, also emptiness for 2ATAs is in \ExpTime.
\end{proof}

\end{document}